\documentclass[%
 reprint,
superscriptaddress,
nofootinbib,
nobibnotes,
 amsmath,amssymb,
 aps,
pra,
]{revtex4-2}

\usepackage{graphicx}
\usepackage{dcolumn}
\usepackage{bm}
\usepackage{hyperref}
\usepackage{tikz-cd}
\usepackage{amsthm}
\usetikzlibrary{arrows.meta}
\tikzcdset{
    arrow style=tikz,
    arrows={>={Circle[open]}}
    }

\newtheorem{theorem}{Theorem}

\newtheorem{definition}{Definition}
\newtheorem{lemma}{Lemma}

\newcommand{\B}{\mathcal{B}}
\newcommand{\density}{\mathcal{D}}
\newcommand{\hil}{\mathcal{H}}

\newcommand{\braket}[2]{\langle #1 | #2 \rangle}
\newcommand{\ketbra}[1][]{| #1 \rangle \langle #1 |}
\newcommand{\ket}[1][]{| #1 \rangle}
\newcommand{\bra}[1][]{\langle #1 |}

\newcommand{\trace}{\mathrm{Tr}}
\newcommand{\id}{\mathbb{I}}
\newcommand{\real}{\mathbb{R}}

\newcommand{\conj}{\theta}

\newcommand{\iden}{\mathrm{id}}

\newcommand{\fidelity}{\mathcal{F}}

\newcommand{\trev}{\Theta}
\newcommand{\adj}{\mathrm{Ad}}

\newcommand{\chan}{\mathcal{E}}
\newcommand{\cat}{\mathbf{C}}
\newcommand{\choi}[1][]{{\mathrm{C}_{ #1 }}}
\newcommand{\cstate}[1][]{{\widetilde{{\mathrm{C}}}_{#1}}}

\newcommand{\swapio}{{\mathcal{S}_{I \leftrightarrow O}}}

\newcommand{\unital}{\textup{\texttt{Unital}}}
\newcommand{\trs}{\textup{\texttt{TRS}}}
\newcommand{\free}{\textup{\texttt{Free}}}

\newcommand{\diag}{{\mathrm{diag}}}

\DeclareMathOperator*{\argmin}{arg\,min}

\begin{document}

\title{Dynamical resource theory of time-reversal symmetry breaking}

\author{Jisho Miyazaki}
\email{jisho.miyazaki@phys.s.u-tokyo.ac.jp}
\affiliation{Department of Physics, Graduate School of Science, The University of Tokyo}
\affiliation{Ritsumeikan University BKC Research Organization of Social Sciences}

\author{Kohdai Kuroiwa}
\affiliation{Institute for Quantum Computing and Department of Combinatorics and Optimization, University of Waterloo}
\affiliation{Perimeter Institute for Theoretical Physics}

\author{Mio Murao}
\affiliation{Department of Physics, Graduate School of Science, The University of Tokyo}
\affiliation{Trans-scale Quantum Science Institute, The University of Tokyo}


\date{\today}

\begin{abstract}
    Time-reversal symmetry and its violation (T-violation) are fundamental across diverse physics domains, from particle physics to fluctuation theorems and reciprocity.
    While time-reversal symmetry for isolated systems is simply determined by the Hamiltonian, quantifying the magnitude of T-violation, especially in noisy quantum processes, has been a subject of ongoing debate. Here, we propose an operationally meaningful framework to quantify the intrinsic T-violation of quantum channels. We integrate operationally defined time-reversal transformations into dynamical resource theory to formulate T-violation as a resource. T-violation decomposes into two distinct components: kinematic T-violation originating from antiunitary motion-reversal, and nonunitality driven purely by thermodynamics. This decomposition resolves the longstanding confusion between broken time-reversal symmetry and non-invertibility. Finally, we analyze the resulting resource theory of kinematic T-violation, showing that it is neither quality-like nor quantity-like, and proving the existence of a universal golden unit even within classical channels.
\end{abstract}


\maketitle

\section{Introduction}
Violation of time-reversal symmetry (T-violation) is often confused with the non-invertibility of physical processes, despite their being fundamentally distinct phenomena.
Time-reversal symmetry in classical dynamics requires both the original and motion-reversed processes to satisfy the equations of motion; this concept generalizes to quantum Hamiltonians through antiunitary motion-reversal operators \cite{wigner1931group,sakurai_modern_2011}. In addition to its fundamental role in particle physics, time-reversal symmetry underlies fluctuation theorems \cite{crooks1999entropy,tasaki_jarzynski_2000,crooks_quantum_2008,aurell_time_2015,manzano_nonequilibrium_2015}, Onsager reciprocity \cite{onsager1931reciprocal1,onsager1931reciprocal2,tsang_quantum_2025}, and other reciprocity relations \cite{Potton_2004_reciprocity,deak_reciprocity_2012,barzanjeh_nonreciprocity_2025}.

In contrast, non-invertibility---typically illustrated by a shattering wine glass or the mixing of coffee and milk---refers to a system's inability to return to its original state because of information loss. Whereas T-violation is a kinematic property, non-invertibility is thermodynamic. Indeed, entropy production is sometimes referred to as thermodynamic violation of time-reversal symmetry \cite{ghosal_identification_2026,liang_thermodynamic_2023,van_vliet_time-reversal_2012}.
Thermodynamic and kinematic T-violations are nevertheless distinct phenomena: one can easily conceive of an invertible physical process that violates kinematic T-symmetry. 

The distinction between T-violation and non-invertibility becomes more subtle for non-invertible state transformations. In this article, we introduce an operational definition of T-symmetry and T-violation for general quantum channels. Shifting from Hamiltonians to quantum channels is essential to account for non-invertibility. We demonstrate that while non-invertibility partially contributes to T-violation, its contribution can be rigorously decoupled from kinematic T-violation.

We formulate the equivalence between original and motion-reversed processes operationally and introduce intrinsic measures of T-violation for quantum channels, drawing on recent developments in dynamical resource theory \cite{gour_how_2019,li_quantifying_2020,liu_operational_2020,gonda_monotones_2023,ji_entropic_2023}. While a unified framework exists for quantifying intrinsic asymmetries of quantum states as resources \cite{gour_resource_2008,gour_measuring_2009,marvian_extending_2014}, T-violation---originally defined at the level of Hamiltonians---is inherently a property of state transformations rather than states themselves. We therefore extend this framework to quantum channel asymmetries using higher-order quantum computation \cite{chiribella_quantum_2008,chiribella_transforming_2008,chiribella_theoretical_2009,taranto2025higher}.

Consequently, our resource theory of T-violation is fundamentally distinct from reference-frame resource theories for breaking superselection under time reversal \cite{Gour_timereversal_2009} and CPT \cite{Skotiniotis_CPT_2013,Skotiniotis_CPT_2014pra}, which characterize the intrinsic asymmetry of states. Their extensions to dynamical resources do not coincide with our notion of T-violation, nor does our dynamical resource straightforwardly reduce to a state counterpart. Unlike these previous approaches, our formulation is consistent with the convention that a T-symmetric Hamiltonian commutes with antiunitary motion reversal, thereby characterizing T-violation directly.

Quantifying the degree of T-violation, rather than merely detecting its presence, is essential across physics. For instance, nonreciprocal devices require strong directionality rather than subtle imbalances \cite{barzanjeh_nonreciprocity_2025}. In particle physics, the magnitude of T-violation is related to parameters of the Cabibbo--Kobayashi--Maskawa matrix \cite{cabibbo1963unitary,kobayashi1973cp} and to the Jarlskog invariant \cite{Jarlskog_1985PhysRevLett.55.1039,jarlskog1985basis,Wu1986_PhysRevD.33.860}. Although our present formulation focuses on quantum channels, it provides a foundation for quantifying kinematic T-violation more generally.

Ultimately, we separate kinematic T-violation from the contribution associated with non-invertibility in general quantum channels. Non-invertible channels may be either unital or nonunital. Unital channels do not decrease system entropy \cite{chiribella_microcanonical_2017,gour_resource_2015} and make no contribution to T-violation. Nonunital channels map the uniform state to a nonuniform state and thereby contribute to T-violation. Isolating this thermodynamic contribution yields the quantitative relation
\begin{equation}
    \label{concept} \text{T-violation} = \text{nonunitality} ~+~ \text{kinematic T-violation}.
\end{equation}
This relation resolves the longstanding confusion between kinematic T-violation and non-invertibility.

Finally, we highlight several properties of our kinematic T-violation resource theory. Reference~\cite{gour_resource_2008} notes that discrete symmetry violations in quantum states do not scale extensively with system size. Although this observation does not apply directly to channel symmetries, we show that kinematic T-violation is neither extensive nor intensive. Furthermore, we construct a classical channel that maximally violates T-symmetry, showing that kinematic T-violation is not unique to quantum systems.

This article is organized as follows. In Sec.~\ref{sec:operational_theory}, after reviewing the physical interpretation of motion-reversal operators and T-symmetry in isolated systems, we define T-symmetry for quantum channels. Section~\ref{sec:nested_classes} provides the necessary ingredients to formulate T-violation as a dynamical resource, including unital and T-symmetric superchannels, and presents the structural decomposition of nonunitality and kinematic T-violation. Section~\ref{sec:decomposition} establishes the quantitative decomposition of these two contributions by applying a generalized Pythagorean theorem in information geometry to the T-violation resource monotones. In Sec.~\ref{sec:kinematic}, we focus on kinematic T-violation and explore its resource-theoretic properties. Finally, Sec.~\ref{sec:conclusion} concludes the paper.

\section{Operational formulation of T-symmetry}\label{sec:operational_theory}
Formally, the resource theory developed in this article may be constructed from any antiunitary operator. Its interpretation as a resource theory of T-violation, however, requires that this antiunitary be assigned the physical role of reversing the motion of the quantum system. This section clarifies the meaning of the antiunitary motion-reversal operator and explains how it serves as the organizing principle for the resource-theoretic formulation.

\subsection{Motion-reversal operator}\label{subsec:motion-reversal}
Throughout this work, the motion-reversal operator is treated as part of the kinematic specification of each quantum system. To clarify the physical content of this assumption, we recall the following standard construction.

The antiunitary motion-reversal operator was introduced by Wigner \cite{wigner1931group} in relation to a preassigned unitary time-displacement operator $U_{\Delta t}$. Given such a notion of time displacement on the system, Wigner characterized motion reversal by an antiunitary operator $\trev$ satisfying
\begin{equation}
    \label{wigner}    U_{\Delta t} \trev U_{\Delta t} \trev = \id, \quad (\forall \Delta t \in \real)
\end{equation}
which expresses the following operational cycle: after reversing the motion, displacing the system in time, reversing the motion once more, and applying the same time displacement again, the system returns to its initial state.

This characterization emphasizes that motion reversal is not an intrinsic structure independent of dynamics; rather, it is defined relative to the chosen time displacement. Its role is to invert the direction of the motion specified by that displacement when acting on quantum states.

Although the term ``time reversal'' is more common, we use ``motion reversal'' deliberately. For the present formulation, it is essential to distinguish the antiunitary operation from a literal reversal of time itself: the operation reverses the kinematic sense of motion, while the interpretation as time reversal rests on how that motion is physically understood. This terminology is also supported by several discussions in the literature, which point out that ``time reversal'' can be misleading and that ``motion reversal'' is often the more precise description \cite{ballentine1998quantum,Gibson:1976wp,lopez_review_2024,sakurai_modern_2011}. Wigner likewise described the operation as a reversal of motion rather than a reversal of time \cite{wigner1931group}.

In this article, the relevant time displacement is not identified with the actual evolution of the system. Instead, we assume a reference, or free, Hamiltonian $H_\mathrm{free}$ and take $U_{\Delta t}=e^{-i H_\mathrm{free} \Delta t}$ as the corresponding free evolution. The motion-reversal operator is then defined relative to this reference displacement. Since condition \eqref{wigner} does not determine $\trev$ uniquely, one typically supplements it with additional physical requirements, such as $\trev \hat{p} \trev^{-1} = - \hat{p}$ for the momentum operator $\hat{p}$. For instance, \cite{Roberts_2017} derives the motion-reversal operator from further constraints, including spatial isotropy.

Different choices of time displacement in \eqref{wigner} therefore lead to different motion-reversal operators and, consequently, to different assessments of T-symmetry \cite{lopez_review_2024}. For example, if the time displacement is identified with the actual evolution $e^{- i H \Delta t}$, then \eqref{wigner} can always be satisfied by choosing an antiunitary operator $\trev$ that commutes with $H$. In that convention, \eqref{wigner} effectively becomes the definition of T-symmetry, and the symmetry is built into the definition of the motion-reversal operator itself. We adopt a different convention: the actual evolution of the system need not coincide with the reference time displacement used to define motion reversal.

Irrespective of the displacement operator, we assume that every motion-reversal operator is \textit{involutive} in the sense that
\begin{equation}
    \trev^2 = e^{i \phi} \id,
\end{equation}
for some phase $\phi$. Physically, this means that applying motion reversal twice returns a system to its original state, consistent with the meaning of ``reversal.'' More generally, if the system has several superselection sectors, the phase factor may depend on the sector. Here, we focus on the simplest setting, without a decomposition into superselection sectors.

\subsection{Traditional formulation of T-symmetry for Hamiltonians}
Once a motion-reversal operator has been specified, one can ask how T-symmetry should be attributed to a quantum system. The most direct possibility would be to call a quantum \emph{state} $\ket[\psi]$ T-symmetric when
\begin{equation}
    \trev \ket[\psi] = \ket[\psi].
\end{equation}
This perspective leads to resource theories of breaking time-reversal superselection \cite{Gour_timereversal_2009}  and imaginarity \cite{hickey_quantifying_2018,wu_resource_2021,wu_operational_2021}, with the reference basis determined by the chosen antiunitary operator. Our objective is different: rather than treating T-violation as a property of states, we seek a formulation grounded in the reversal of dynamical processes.

The relevant intuition comes from classical mechanics. Dynamics carrying an initial state to a final state are regarded as T-symmetric when the corresponding motion-reversed dynamics also satisfy the equation of motion. The motion-reversed trajectory begins at the motion-reversed final state and terminates at the motion-reversed initial state.

In quantum mechanics, the conventional formulation of this idea is as follows \cite{sakurai_modern_2011}: whenever $\ket[\psi_s(t)]$ is a solution of the Schr\"odinger equation
\begin{equation}
    \label{schrodinger}    i \frac{d}{dt} \ket[\psi (t)] = H \ket[\psi (t)],
\end{equation}
then $\trev \ket[\psi_s(-t)]$ is also a solution. This condition holds if and only if
\begin{equation}
    \label{hamiltonian} [\trev, H] = 0.
\end{equation}

From an operational standpoint, imposing T-symmetry directly at the level of the Schr\"odinger equation is unnecessarily stringent, because it constrains unobservable global phases. Nevertheless, the same Hamiltonian condition is recovered even if one instead starts from the Liouville--von Neumann equation
\begin{equation}
    i \frac{d}{dt} \rho(t) = [H,\rho(t)],
\end{equation}
as the equation of motion. Indeed, suppose that $\rho(t) := e^{-iHt} \rho_0 e^{iHt}$ is a solution, and require the motion-reversed trajectory $\rho_\trev(t) := \trev \rho(-t) \trev^\dagger$ to be a solution as well. Then
\begin{eqnarray*}
    [H,\rho_\trev(t)] &= i \frac{d}{dt} \rho_\trev(t) = i \trev \left( \frac{d}{dt} \rho(-t) \right) \trev^\dagger \\
    &= \trev [H,\rho(-t)] \trev^\dagger = [\trev H \trev^\dagger, \rho_\trev(t)],
\end{eqnarray*}
and hence
\begin{equation}
    [H - \trev H \trev^\dagger, \rho_\trev(t)]=0.
\end{equation}
Since this relation must hold for every possible density matrix $\rho_\trev(t)$, Schur's lemma implies
\begin{equation}
    \label{identity_freedom}    \trev H \trev^\dagger = H + c \id. \quad (c \in \mathbb{C})
\end{equation}
Because $\trev$ is antiunitary, $\trev H \trev^\dagger$ has the same spectrum as $H$. The two sides of \eqref{identity_freedom} can therefore have identical spectra only when $c=0$, which yields \eqref{hamiltonian}. On this basis, we take \eqref{hamiltonian} as the definition of a T-symmetric Hamiltonian.\footnote{Under the alternative convention, discussed at the end of Sec.~\ref{subsec:motion-reversal}, that identifies the reference time displacement with the actual time evolution, T-symmetry is regarded as an intrinsic property of quantum mechanics and is satisfied by construction.}


\subsection{Process-based T-symmetry and microreversibility}
We next introduce a definition of T-symmetry for unitary transformations that more directly reflects the preceding process-based intuition and is better suited to an operational formulation. The resulting criterion is closely related to ``microreversibility,'' which is usually regarded as a consequence of T-symmetry.

We regard the unitary transformation $U$ as the underlying dynamics, while distinguishing it from the observed ``process,'' understood simply as an input-output relation. If the original process is
\begin{equation}
    \ket[\psi_1] \rightarrow \ket[\psi_2],
\end{equation}
then the corresponding motion-reversed process is
\begin{equation}
    \trev \ket[\psi_2] \rightarrow \trev \ket[\psi_1].
\end{equation}
The relation between the two processes is illustrated in Fig.~\ref{fig:motion_reversal}.
\begin{figure}[tbp]
    \includegraphics[width=.45\textwidth]{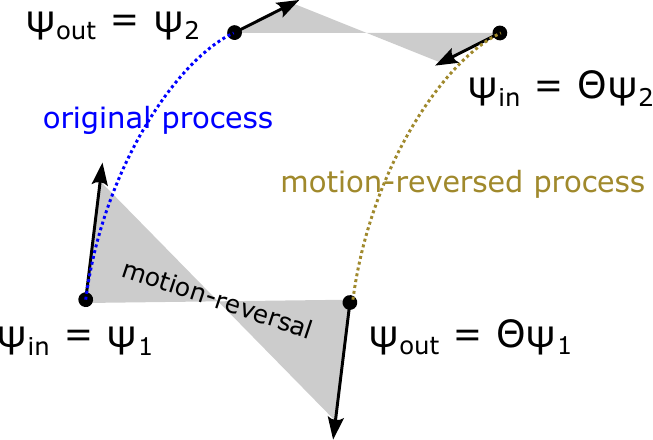}
    \caption{\label{fig:motion_reversal} Schematic relation between original and motion-reversed processes for a system with momentum.}
\end{figure}
T-symmetry is thus expressed at the level of observable processes: the original and motion-reversed processes must occur with the same probability. More precisely,
\begin{definition}\label{def:trs_operational}
    Let $\trev:\hil \rightarrow \hil$ be the motion-reversal antiunitary operator. A unitary transformation $U:\hil \rightarrow \hil$ is defined to be T-symmetric if
    \begin{equation}
        \label{transition_prob} | \langle \psi_2 | U | \psi_1 \rangle | = | \langle \trev \psi_1 | U | \trev \psi_2 \rangle |. \quad (\forall \ket[\psi_1], \ket[\psi_2] \in \hil)
    \end{equation}
\end{definition}
\noindent This notion of T-symmetry is closely related to microreversibility. A formulation of microreversibility particularly close to Def.~\ref{def:trs_operational} appears, for example, in the discussion of scattering matrices in \cite{messiah1961quantum} (see Chapter XIX, $\S$ 3), where original and motion-reversed scattering processes are required to have the same amplitude. Whether microreversibility equates transition amplitudes, including their complex phases, or only transition probabilities, as in Def.~\ref{def:trs_operational}, depends on the context.


The following theorem relates this operational definition to the Hamiltonian condition \eqref{hamiltonian}.
\begin{theorem}\label{thm:trs_algebraic}
    \begin{description}
        \item[(i)] A unitary transformation is T-symmetric if and only if there exists a phase factor $e^{i \phi}$ such that $U_\trev := \trev^\dagger U^\dagger \trev = e^{i \phi} U$.
        \item[(ii)] The unitary transformation $e^{-iHt}$ is T-symmetric for any $t \in \real$ if and only if $[H,\trev]=0$.
        In this case, $\trev^\dagger e^{iHt} \trev = e^{-iHt}$ holds for any $t$.
    \end{description}
\end{theorem}
\begin{proof}
    (i) Recall that the Hermitian adjoint of antilinear operators is defined by $\bra[\psi] (\trev^\dagger \ket[\phi]) = \bra[\phi] (\trev \ket[\psi])$. We have
    \begin{eqnarray}
        \langle \trev \psi_1 | U | \trev \psi_2 \rangle &=  \langle U^\dagger \trev \psi_1 | \trev \psi_2 \rangle \\
        &= \langle \psi_2 |  \trev^\dagger | U^\dagger \trev \psi_1 \rangle \\
        &= \langle \psi_2 |  \trev^\dagger U^\dagger \trev | \psi_1 \rangle.
    \end{eqnarray}
    The T-symmetry constraint is thus equivalent to
    \begin{equation}
        | \langle \psi_2 | U | \psi_1 \rangle | = | \langle \psi_2 |  U_\trev | \psi_1 \rangle |. \quad (\forall \ket[\psi_1], \ket[\psi_2] \in \hil)
    \end{equation}
    This equality states that the unitary channels $U \cdot U^\dagger$ and $U_\trev \cdot U_\trev^\dagger$ coincide, which means that the operators $U_\trev$ and $U$ differ only by a phase.

    (ii) Since sufficiency is immediate, we prove necessity. Assume
    \begin{equation}
        \trev^\dagger e^{iHt} \trev = e^{i \phi(t)} e^{-iHt},
    \end{equation}
    with a $t$-dependent phase $\phi(t)$. Both sides are unitary. Differentiating with respect to $t$ gives
    \begin{eqnarray}
        \nonumber & -i \trev^\dagger H \trev (\trev^\dagger e^{iHt} \trev) = (i \phi(t)' \id -i H) (e^{i \phi(t)} e^{-iHt}) \\
        \Leftrightarrow & \trev^\dagger H \trev = H - \phi(t)' \id.
    \end{eqnarray}
    For the spectra of both sides to agree, one must have $\phi(t)' = 0$. Hence $[H,\trev]=0$, and $\phi(t)$ is constant. This constant can be taken to be zero, since at $t=0$ we have $\id = \trev^\dagger \id \trev = e^{i \phi} \id$.
\end{proof}
Definition~\ref{def:trs_operational} therefore gives an operational test for T-violation: one compares the transition probabilities of the original and motion-reversed processes.

Although the T-symmetry condition may be identified with microreversibility, we avoid doing so for several reasons. First, particularly in the context of fluctuation theorems, microreversibility is often imposed only on a restricted class of transitions, such as transitions between energy eigenstates, whereas Def.~\ref{def:trs_operational} requires the condition to hold for every pair of states. When microreversibility is imposed only on a restricted class of states, its equivalence to Hamiltonian T-symmetry, established in Theorem~\ref{thm:trs_algebraic}, need not follow.
Second, although T-symmetry and microreversibility may be identified, our notion of T-\textit{violation} is not readily identified with conventional micro\textit{ir}reversibility. The resource theory developed here is motivated by kinematic T-violation.
By contrast, violations of microreversibility (and hence detailed balance) have primarily been associated with thermodynamic T-violation in previous studies \cite{ghosal_identification_2026,liang_thermodynamic_2023,van_vliet_time-reversal_2012}.
For terminological consistency, we therefore use ``T-symmetry'' in Def.~\ref{def:trs_operational}.

We now assign terminology to the unitary operator $U_\trev$ appearing in Thm.~\ref{thm:trs_algebraic} (i).
\begin{definition}\label{def:time_reversal_unitary}
    For a unitary operator $U:\hil \rightarrow \hil$, the operator $\trev^\dagger U^\dagger \trev$ is referred to as its reverse-process simulator (RPS) and denoted by $U_\trev$.
\end{definition}
\noindent RPS is slightly but strikingly different from the ``time reversal'' of unitary used in fluctuation theorems \cite{campisi_colloquium_2011,aurell_time_2015,manzano_nonequilibrium_2015}, given by
\begin{equation}
    \label{trev_unitary}    U^\mathrm{tr}_\trev := \trev U^\dagger \trev^\dagger.
\end{equation}
The time reversal is defined to satisfy
\begin{equation}
    | \langle \psi_2 | U | \psi_1 \rangle | = | \langle \trev \psi_1 | U^\mathrm{tr}_\trev | \trev \psi_2 \rangle |. \quad (\forall \ket[\psi_1], \ket[\psi_2] \in \hil),
\end{equation}
whereas the RPS satisfies
\begin{equation}
    \label{time_reversal_unitary}   | \langle \psi_2 | U_\trev | \psi_1 \rangle | = | \langle \trev \psi_1 | U | \trev \psi_2 \rangle |. \quad (\forall \ket[\psi_1], \ket[\psi_2] \in \hil).
\end{equation}
Equation~\eqref{time_reversal_unitary} shows that the transition probability of the motion-reversed process $\trev \ket[\psi_2] \rightarrow \trev \ket[\psi_1]$ can be represented as the transition probability of the original process $\ket[\psi_1] \rightarrow \ket[\psi_2]$ under $U_\trev$. This motivates calling $U_\trev$ the reverse process simulator (RPS) of $U$.

Before extending RPS to quantum channels, we record two elementary features of the unitary case. First, any antiunitary operator can be decomposed as $\trev = \conj X$, where $\conj$ is complex conjugation in the computational basis and $X$ is a unitary operator. The RPS is then represented as
\begin{equation}
    U_\trev = X^\dagger U^\top X,
\end{equation}
where $\top$ denotes transposition. Thus, RPS amounts to transposition up to unitary conjugation.

Second, symmetric unitaries are not generally closed under sequential composition. Indeed,
\begin{equation}
    (UV)_\trev = \trev^\dagger V^\dagger U^\dagger \trev = V_\trev U_\trev \neq U_\trev V_\trev = UV.
\end{equation}
Thus, noncommutativity of $U$ and $V$ obstructs closure under sequential composition.

\subsection{T-symmetry of channels}\label{subsec:trs_channels}
We now extend the preceding operational definition to general dynamics, represented in operational quantum theory by quantum channels. This extension is necessary for treating experimentally relevant situations in which noise cannot be neglected, and it is also essential for formulating a resource theory in terms of the broadest operationally available transformations. The guiding principle is the same as in Def.~\ref{def:trs_operational}: dynamics are T-symmetric when the original and motion-reversed processes occur with the same probability.

Generalizing transition probabilities requires particular care. A naive extension of condition \eqref{transition_prob} to quantum channels is
\begin{equation}
    \label{trs_channel_condition0}   \trace[\rho_2 \chan (\rho_1)] = \trace[\trev \rho_1 \trev^\dagger \chan (\trev \rho_2 \trev^\dagger)], \quad \forall \rho_1, \rho_2.
\end{equation}
However, the left-hand side is usually interpreted as the probability of the effect $\rho_2$ on the state $\chan (\rho_1)$, and the right-hand side as that of the effect $\trev \rho_1 \trev^\dagger$ on the state $\chan (\trev \rho_2 \trev^\dagger)$, rather than as transition probabilities between states. This interpretation requires $\trev$ to interchange states and effects, an assumption that lacks an immediate operational justification. This issue does not arise for spatial symmetry transformations such as rotations and parity inversion, which map states to states and effects to effects.

Oreshkov and Cerf \cite{oreshkov_operational_2015} provide an operational formulation of time reversal that interchanges state preparations and measurements. We follow their approach to define T-symmetry for channels. Their formulation generalizes the Born rule to treat preselection in state preparation and postselection in measurement symmetrically. A state preparation on system $A$ is represented by a pair of positive semidefinite operators $(\rho_A ; \overline{\rho}_A)$ such that $\rho_A \leq \overline{\rho}_A$ and $\trace[\overline{\rho}_A]=1$. An effect on $A$ is represented by a pair of positive semidefinite operators $(E_A ; \overline{E}_A)$ such that $E_A \leq \overline{E}_A$ and $\trace[\overline{E}_A] = d_A$. The generalized probability of the event $(\rho_A,E_A)$, conditioned on the preparation and postselection specified by $\overline{\rho}_A$ and $\overline{E}_A$, is
\begin{equation}
    p((\rho_A ; \overline{\rho}_A),(E_A ; \overline{E}_A)) = \frac{\trace[E_A \rho_A]}{\trace[\overline{E}_A \overline{\rho}_A]},
\end{equation}
whenever the denominator is nonzero; otherwise, $p((\rho_A ; \overline{\rho}_A),(E_A ; \overline{E}_A)) = 0$. The denominator $\trace[\overline{E}_A \overline{\rho}_A]$ normalizes the probability over the restricted class of events. The standard Born rule is recovered by choosing $\rho_A = \overline{\rho}_A$ and $\overline{E}_A = \id_A$. In the Oreshkov--Cerf formulation, a time-reversal transformation maps states to effects and effects to states. Before specifying the motion-reversal operator, we write the general transformation $\hat{S}$ allowed in operational probabilistic theory as
\begin{align}
    \label{oreshkov_cerf_1} \hat{S}_{s \rightarrow e} (\rho_A ; \overline{\rho}_A) &= \left( d_A \frac{S \rho_A^\top S^\dagger}{\trace[S \overline{\rho}_A^\top S^\dagger ]} ; d_A \frac{S \overline{\rho}_A^\top S^\dagger}{\trace[S \overline{\rho}_A^\top S^\dagger ]} \right) \\
    \label{oreshkov_cerf_2} \hat{S}_{e \rightarrow s} (E_A ; \overline{E}_A) &= \left( \frac{S^{-1 \dagger} E_A^\top S^{-1}}{\trace[S^{-1 \dagger} \overline{E}_A^\top S^{-1}]} ; \frac{S^{-1 \dagger} \overline{E}_A^\top S^{-1}}{\trace[S^{-1 \dagger} \overline{E}_A^\top S^{-1}]} \right),
\end{align}
where $S$ is any invertible linear operator on $\hil_A$.

To define T-symmetric channels within the Oreshkov--Cerf formulation, we specialize the transformations in \eqref{oreshkov_cerf_1} and \eqref{oreshkov_cerf_2} using the motion-reversal operator $\trev$ on $A$:
\begin{align}
    \label{s_to_e}    \hat{\trev}_{s \rightarrow e} (\rho_A ; \overline{\rho}_A) &= \left( d_A \frac{\trev \rho_A \trev^\dagger}{\trace[\trev \overline{\rho}_A \trev^\dagger ]} ; d_A \frac{\trev \overline{\rho}_A \trev^\dagger}{\trace[\trev \overline{\rho}_A \trev^\dagger ]} \right) \\
    &= \left( d_A \trev \rho_A \trev^\dagger ; d_A \trev \overline{\rho}_A \trev^\dagger \right), \\
    \label{e_to_s}  \hat{\trev}_{e \rightarrow s} (E_A ; \overline{E}_A) &= \left( \frac{\trev E_A \trev^{-1}}{\trace[\trev \overline{E}_A \trev^{-1}]} ; \frac{\trev \overline{E}_A \trev^{-1}}{\trace[\trev \overline{E}_A \trev^{-1}]} \right) \\
    &= \left( \frac{1}{d_A} \trev E_A \trev^\dagger ; \frac{1}{d_A}\trev \overline{E}_A \trev^\dagger \right),
\end{align}
by taking $\trev = S \conj$. This kinematic specification is not implied by the structure of operational probabilistic theory alone. The composition rule of \cite{oreshkov_operational_2015} generalizes the transition probability $|\langle \psi_2 | U | \psi_1 \rangle|^2$ to channels as
\begin{equation}
    p((\rho_A ; \overline{\rho}_A),\chan,(E_A ; \overline{E}_A)) := \frac{\trace[E_A \chan(\rho_A)]}{\trace[\overline{E}_A \chan(\overline{\rho}_A)]}.
\end{equation}
With these preliminaries, we introduce an operational definition of T-symmetric channels.
\begin{definition}
    Let $\trev: \hil \rightarrow \hil$ be the motion-reversal antiunitary operator. A quantum channel (i.e., a CPTP map) $\chan:\B(\hil) \rightarrow \B(\hil)$ is defined to be T-symmetric if
    \begin{equation}
        \label{trs_channel_condition}   p((\rho ; \overline{\rho}),\chan,(E ; \overline{E})) = p( \hat{\trev}_{e \rightarrow s} (E ; \overline{E}),\chan,\hat{\trev}_{s \rightarrow e} (\rho ; \overline{\rho}))
    \end{equation}
    for any pair consisting of a state $(\rho ; \overline{\rho})$ and an effect $(E ; \overline{E})$ on $\hil$ satisfying
    \begin{equation}
        \label{selection_symmetry}  (\dim \hil) \trev \overline{\rho} \trev^\dagger = \overline{E}.
    \end{equation}
\end{definition}
\noindent For this definition to be meaningful, the input and output spaces of the channel must be identical; otherwise, the effect $(E ; \overline{E})$, which belongs to the output space on the left-hand side of \eqref{trs_channel_condition}, cannot be motion-reversed into an input state $\hat{\trev}(E ; \overline{E})$ on the right-hand side.
The condition \eqref{selection_symmetry} requires the original and motion-reversed processes to share the same preselection and postselection scenarios.

The operation of taking the RPS can likewise be generalized from unitaries to channels.
\begin{definition}\label{def:time_reversal_channel}
    If $\chan:\B(\hil) \rightarrow \B(\hil)$ is a quantum channel, its reverse-process simulator (RPS) is defined to be the CP map $\chan_\trev$ such that
    \begin{equation}
        \label{time_reversal_interpretation}    p( \hat{\trev}_{e \rightarrow s} (E ; \overline{E}),\chan,\hat{\trev}_{s \rightarrow e} (\rho ; \overline{\rho})) = p((\rho ; \overline{\rho}),\chan_\trev,(E ; \overline{E})),
    \end{equation}
    for any pair consisting of a state $(\rho ; \overline{\rho})$ and an effect $(E ; \overline{E})$ on $\hil$ satisfying
    \begin{equation}
        \label{selection_symmetry_reversal}  (\dim \hil) \trev \overline{\rho} \trev^\dagger = \overline{E}.
    \end{equation}
\end{definition}
\noindent The RPS encodes the transition probability of the motion-reversed process in the language of the original process. The following theorem is an immediate consequence of the definition.
\begin{theorem}\label{thm:trs_channel}
    The reverse-process simulator (RPS) of a channel $\chan$ is given by
    \begin{equation}
        \chan_\trev := \adj_{\trev^\dagger} \circ \chan^\dagger \circ \adj_\trev = \adj_{X^\dagger} \circ \chan^\top \circ \adj_X,
    \end{equation}
    where $\adj_O$ denotes the adjoint action $\adj_O (\cdot) := O \cdot O^\dagger$, $\chan^\dagger$ denotes the Hilbert--Schmidt adjoint of $\chan$, and $\trev = \conj X$.
    A quantum channel is T-symmetric if and only if it is equal to its RPS, i.e., $\chan = \chan_\trev$.
\end{theorem}
\noindent Importantly, the adjoint $\chan^\dagger$, and therefore the RPS $\chan_\trev$, need not be a quantum channel, because it may fail to be trace-preserving.

Again, RPS is different from the time reversal of channels, given by
\begin{equation}
    \chan^\mathrm{tr}_\trev := \adj_\trev \circ \chan^\dagger \circ \adj_{\trev^\dagger},
\end{equation}
which reduces to \eqref{trev_unitary} for unitary channels.
When states and effects are reversed according to \eqref{s_to_e} and \eqref{e_to_s},
the time reversal of channels that appear in the Oreshkov--Cerf formulation coincides with $\chan^\mathrm{tr}_\trev$.
In their formulation, the time reversal of an operation from system $A_\mathrm{in}$ to system $A_\mathrm{out}$ is an operation from $A_\mathrm{out}$ to $A_\mathrm{in}$ and therefore reverses the temporal order. By contrast, the RPS channel in \eqref{time_reversal_interpretation} has the same temporal orientation as the original channel: its action on the original process reproduces the probabilities of the motion-reversed process.
Since a channel and its RPS share the same temporal orientation, equating them as in Theorem~\ref{thm:trs_channel} makes sense and hence the notion of kinematic T-symmetry is well defined.

As a symmetry transformation, the mapping $\chan \rightarrow \chan_\trev$ is not realizable as a linear supermap \cite{chiribella_symmetry_2021}. Instead, it belongs to the class known as ``input-output inversion'' \cite{chiribella_quantum_2022}, whose prominent feature is reversing the order of sequential composition, i.e., $(\chan^2 \circ \chan^1 )_\trev = \chan^1_\trev \circ \chan^2_\trev$. For this reason, one might be tempted to consider the input and output of $\chan_\trev$ to be reversed relative to those of the original channel $\chan$. However, as mentioned in the previous paragraph, we insist on a shared temporal orientation for both $\chan_\trev$ and $\chan$, which enables a direct comparison between them. This convention remains consistent as long as sequential compositions of the form $\chan^2 \circ \chan^1$ are not involved.

The Choi-operator representation of the RPS channel will play an important role in the subsequent analysis. A direct calculation gives
\begin{equation}
    \label{choi_time_reversal}  \choi[\chan_\trev] = \adj_{X^\dagger_\mathrm{out} \otimes X^\top_\mathrm{in}} \circ \swapio (\choi[\chan]),
\end{equation}
where $\trev = \conj X$, $\choi[\chan]$ denotes the Choi operator of $\chan$, and $\swapio$ denotes the operation that interchanges the input and output spaces of the Choi operator.

Two features arise here that have no analogue for unitary transformations. First, T-symmetry in this sense is independent of channel invertibility. The completely depolarizing channel $\chan: \rho \mapsto \id / d$, although maximally non-invertible, is T-symmetric. This is because the RPS of a channel is defined through the Hilbert--Schmidt adjoint $\chan^\dagger$, rather than through an inverse map $\chan^{-1}$. Theorem~\ref{thm:trs_channel} therefore makes it possible to isolate T-violation that persists even in the presence of noise. This highlights the distinction between our T-symmetry and thermodynamic T-symmetry that is violated by entropy production.

Second, channels exhibit an additional mechanism for T-violation that has no unitary counterpart: nonunitality. If the adjoint channel is not trace-preserving, then the RPS cannot coincide with a quantum channel. Channels whose adjoints are also channels are precisely the unital channels, also called bistochastic or bidirectional channels \cite{chiribella_quantum_2022}, characterized by $\chan (\id) = \id$. Hence, any channel that maps the uniform input $\id/d$ to a nonuniform output necessarily breaks T-symmetry.

The preceding argument suggests that overall T-violation comprises both nonunitality and ``kinematic'' T-violation. This intuition can be formulated rigorously using resource theory and quantum information geometry. In Sec.~\ref{sec:decomposition}, we define resource monotones for these three concepts and relate them through a simple identity.

\section{Nested classes of free superchannels}\label{sec:nested_classes}
A standard formulation of a quantum resource theory begins by specifying the free states and free operations, namely those objects and transformations that do not contain the resource. This specification underlies both the associated convertibility problem and operational quantifications of the resource.

The form of T-violation considered here does not fit neatly into this standard state-based framework. First, the resource of interest is intrinsic to channels rather than states. Resource theories of quantum channels, or dynamical resource theories, have been studied extensively \cite{gour_how_2019,li_quantifying_2020,liu_operational_2020,gonda_monotones_2023,ji_entropic_2023}, but several of them are constructed from an underlying state resource theory, whereas no such state-level resource theory is available for the present formulation of T-violation. Finally, for kinematic T-violation, the ambient class of admissible processes must be restricted to unital quantum channels, which form a proper subclass of all quantum channels.

To address these issues, we use the ``top-down'' approach of \cite{ji_entropic_2023}, which starts from superchannels rather than states, to construct nested dynamical resource theories.
We define unital and T-symmetric superchannels as subclasses of deterministic superchannels. The nested resource theories of T-violation, nonunitality, and kinematic T-violation are defined by the respective classes of channels and superchannels in Table~\ref{tab:nest}.
\begin{table}[htbp]
\caption{\label{tab:nest} Admissible channels, free channels and free superchannels of three dynamical resource theories.}
    \begin{ruledtabular}
        \begin{tabular}{lccc}
            Resource &
            \begin{tabular}{c}
                Admissible \\ channels 
            \end{tabular} &
            \begin{tabular}{c}
                Free \\ channels
            \end{tabular} &
            \begin{tabular}{c}
                Free \\ superchannels
            \end{tabular} \\
            \colrule
            T-violation & arbitrary & T-symmetric & T-symmetric \\
            Nonunitality & arbitrary & unital & unital \\
            \begin{tabular}{c}
                Kinematic ~~~ \\ ~ T-violation
            \end{tabular} & unital & T-symmetric & T-symmetric \\
        \end{tabular}
    \end{ruledtabular}
\end{table}
As shown below, every T-symmetric superchannel is unital. This inclusion is essential for defining kinematic T-violation.

\subsection{Notation}
We label systems by capital Roman letters $A,B,\ldots$. The special ``void'' system, represented by the one-dimensional space $\mathbb{C}$, is denoted by $I$.

The sets of unital and T-symmetric channels on $A$ are denoted by $\unital(A)$ and $\trs(A)$, respectively. When considering time-reversal symmetry, we assume that system $A$ is specified by both its Hilbert space $\hil_A$ and its motion-reversal operator $\trev_A$. The latter kinematic specification is unnecessary when considering nonunitality alone.

Deterministic superchannels (or simply superchannels), also known as deterministic combs or networks \cite{chiribella_quantum_2008,chiribella_theoretical_2009,chiribella_transforming_2008}, are transformations from channels to channels that can be physically implemented by encoder--decoder channel pairs.\footnote{All superchannels considered in this article have definite causal order.}

Mathematically, a superchannel from channels on $A$ to channels on $B$ is a specific linear map from $\B(\B(\hil_A),\B(\hil_A))$ to $\B(\B(\hil_B),\B(\hil_B))$. We denote the input and output spaces of channels on $A$ by $\hil_A^\mathrm{in}$ and $\hil_A^\mathrm{out}$, respectively, and use the analogous notation for $B$. In this notation, a superchannel $f$ from $A$ to $B$ is represented by a CPTP map from $\hil_B^\mathrm{in} \otimes \hil_A^\mathrm{out}$ to $\hil_A^\mathrm{in} \otimes \hil_B^\mathrm{out}$ with the no-signaling condition from $\hil_A^\mathrm{out}$ to $\hil_A^\mathrm{in}$.  We denote a superchannel $f$ of this type simply by
\begin{equation}
    f:A \multimap B,
\end{equation}
the corresponding CPTP map (refered to as ``quantum staircase'' in \cite{yokojima2026probabilisticstorageretrievalquantum}) by
\begin{equation}
    \mathcal{M}_f: \B(\hil_B^\mathrm{in} \otimes \hil_A^\mathrm{out}) \rightarrow \B(\hil_A^\mathrm{in} \otimes \hil_B^\mathrm{out}),
\end{equation}
and its Choi operator by $\choi[f] (= \choi[\mathcal{M}_f])$. Figure~\ref{fig:staircase} visualizes the relation between superchannel $f$ and corresponding CPTP map $\mathcal{M}_f$.
\begin{figure}[tbp]
    \includegraphics[width=.45\textwidth]{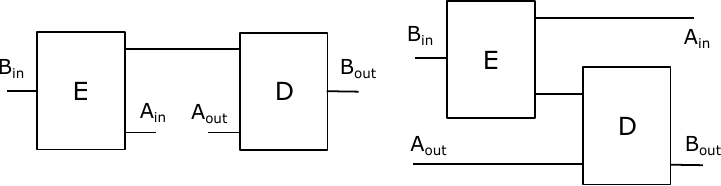}
    \caption{\label{fig:staircase} (Left) the circuit representation of a general superchannel $f:A \multimap B$ and (right) its map form $\mathcal{M}_f$.}
\end{figure}

The sequential composition of superchannels $f:A \multimap B$ and $g:B \multimap C$ is denoted by $g \bullet f: A \multimap C$, to distinguish it from the usual composition of channels $\circ$. Note that the ``sequence'' of superchannels does not follow the temporal order of operations. The parallel composition of superchannels is simply denoted by $\otimes$. The parallel composition of $f_1:A_1 \multimap B_1$ and $f_2:A_2 \multimap B_2$ has the type $A_1 \otimes A_2 \multimap B_1 \otimes B_2$.

For every system $E$, there is a unique identity superchannel that leaves every input channel unchanged. We denote it by $i_E: E \multimap E$.

\subsection{Unital superchannels}\label{subsec:unital}
We first define unital superchannels.
\begin{definition}\label{def:unital_superchannels}
    A deterministic superchannel $f:A \multimap B$ is said to be completely unitality-preserving, or simply unital, if, for every ancillary system $\hil_E$ and every $\chan \in \unital(A \otimes E)$, $f \otimes i_E (\chan) \in \unital(B \otimes E)$, where $i_E$ is the identity superchannel on $E$. The set of all unital superchannels from $A$ to $B$ is denoted by $\unital(A,B)$.
\end{definition}
\noindent Equivalently, a superchannel is unital if it maps unital channels to unital channels even when it acts only on a subsystem of a composite unital channel. This notion is closely related to the higher-order transformations of unital channels studied in \cite{apadula_higher-order_2026}; here, however, we restrict attention to higher-order transformations implementable by standard quantum circuits.

A simple class of unital superchannels is constructed by using unital channels for both the encoder and decoder, as depicted in Figure~\ref{fig:superchannels} (a). This construction relies on the fact that sequential composition preserves unitality. The superchannel of type $A \multimap I$ that discards its input channel (depicted in Figure~\ref{fig:superchannels} (c)) is also unital, even though neither its encoder nor its decoder is unital.
\begin{figure}[tbp]
    \includegraphics[width=.3\textwidth]{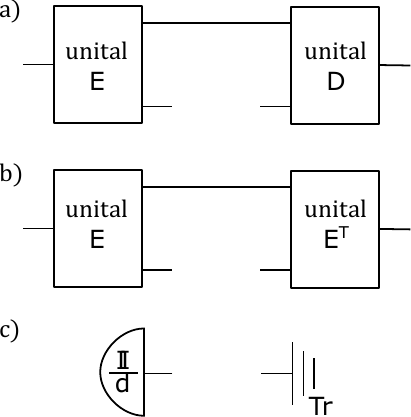}
    \caption{\label{fig:superchannels} (a) Circuit representation of an example of a unital superchannel, featuring a unital encoder channel $E$ and a unital decoder channel $D$. (b) Circuit representation of an example of a T-symmetric superchannel, featuring a unital encoder channel $E$ and its transposed unital decoder channel $E^\top$. The motion reversal is assumed to be given by complex conjugation. (c) A T-symmetric superchannel that discards its input channel.}
\end{figure}

For later use, we extend the definition of adjoint channel to superchannels.
\begin{definition}
    A superchannel $f':A \multimap B$ satisfying
    \begin{equation}
        \label{adjoint_superchannel}    [f \otimes i_E (\chan)]^\dagger = f' \otimes i_E (\chan^\dagger),
    \end{equation}
    for any ancillary system $E$ and any deterministic channel $\chan$ on $A \otimes E$, is called the adjoint of $f$ and is denoted by $f^\dagger$.
\end{definition}
\begin{lemma}\label{lem:unital}
    For any superchannel, its adjoint exists and is unique. For any unital superchannel, its adjoint is also a unital superchannel.
\end{lemma}
\begin{proof}
    We denote the input and output spaces of channels on $A$ by $\hil_A^\mathrm{in}$ and $\hil_A^\mathrm{out}$, respectively, and use the analogous notation for $B$. Let $\mathcal{F}_e$ and $\mathcal{F}_d$ be the encoding and decoding channels of $f$, respectively, sharing a memory system $M$. The left-hand side of \eqref{adjoint_superchannel} can be written as
    \begin{align}
        \nonumber & [(\mathcal{F}_e \otimes \iden_E) \circ (\iden_M \otimes \chan) \circ (\mathcal{F}_d \otimes \iden_E)]^\dagger \\
        & = (\mathcal{F}_d^\dagger \otimes \iden_E) \circ (\iden_M \otimes \chan^\dagger) \circ (\mathcal{F}_e^\dagger \otimes \iden_E).
    \end{align}
    Therefore, defining $f^\dagger:A \multimap B$ as the superchannel with encoder $\mathcal{F}_d^\dagger$ and decoder $\mathcal{F}_e^\dagger$ (which need not be trace-preserving) yields the adjoint of $f$. Its uniqueness follows from linearity.

    If $f:A \multimap B$ is unital, then, for any unital channel $\chan$ on $A \otimes E$, we have
    \begin{equation}
        f^\dagger \otimes i_E (\chan) = [f \otimes i_E (\chan^\dagger)]^\dagger.
    \end{equation}
    Since $\chan^\dagger$ is unital and $f$ is unital, $f \otimes i_E (\chan^\dagger)$ is a unital channel. Therefore, $f^\dagger$ preserves unital channels completely.
\end{proof}

A class of superchannels defines the free transformations of a dynamical resource theory when it contains the identity superchannels and is closed under sequential composition \cite{ji_entropic_2023}. Unital superchannels satisfy both conditions. The free \textit{channels} of this resource theory are the superchannels from the void system $I$, which are precisely the unital channels:
\begin{equation}
    \unital(I,A) = \unital (A).
\end{equation}
Thus, unital superchannels define the dynamical resource theory of \textit{nonunitality}.

Notably, $\unital(A,B)$ is the maximal set of completely resource-nongenerating superchannels for nonunitality. In other words, every superchannel outside $\unital(A,B)$ can convert some unital channel into a nonunital one, possibly when acting on part of a composite system.

\subsection{T-symmetric superchannels}\label{subsec:trs_superchannels}
We extend the process-based definition of RPS in Sec.~\ref{subsec:trs_channels} to superchannels. The guiding principle remains the same: for T-symmetric channels and superchannels, the original and motion-reversed processes should occur with equal probability.

We define RPS of superchannels as follows.
\begin{definition}\label{def:trs_superchannel}
    For a deterministic superchannel $f:A \multimap B$, its reverse-process simulator (RPS) is defined as a superchannel $f_\trev:A \multimap B$, not necessarily deterministic, satisfying
    \begin{align}
        \nonumber &p((\rho_B \otimes \rho_A ; \overline{\rho}_B \otimes \overline{\rho}_A), \mathcal{M}_{f_\trev}, (E_A \otimes E_B ; \overline{E}_A \otimes \overline{E}_B)) = \\
        \label{time_reversal_superchannel} & p((E_B^\trev \otimes E_A^\trev ; \overline{E}_B^\trev \otimes \overline{E}_A^\trev), \mathcal{M}_f, (\rho_A^\trev \otimes \rho_B^\trev ; \overline{\rho}_A^\trev \otimes \overline{\rho}_B^\trev)),
    \end{align}
    for any states $(\rho; \overline{\rho})$ and effects $(E ;\overline{E})$ in the Oreshkov--Cerf formulation, where
    \begin{align}
        &E^\trev := \frac{\trev E \trev^\dagger}{d}, ~ \rho^\trev := d \trev \rho \trev^\dagger,\\
        &\overline{E}^\trev := \frac{\trev \overline{E} \trev^\dagger}{d} = \overline{\rho}, ~ \overline{\rho}^\trev := \frac{\trev \overline{\rho} \trev^\dagger}{d} = \overline{E}.
    \end{align}
    The set of all T-symmetric superchannels from $A$ to $B$ is denoted by $\trs(A,B)$.
\end{definition}
The RPS $f_\trev$, acting on the original process, reproduces the transition probability of the motion-reversed process. The original process comprises processes on $A$ and $B$ with the temporal ordering
\begin{align}
    (\rho_B ; \overline{\rho}_B) &\rightarrow (E_B ; \overline{E}_B), \\
    (E_A ; \overline{E}_A) &\rightarrow (\rho_A ; \overline{\rho}_A).
\end{align}
\begin{lemma}
    For any superchannel $f:A \multimap B$, its RPS exists and is uniquely determined by
    \begin{equation}
        \label{tr_superchannel} f_\trev = \mathcal{A}_{\trev_B} \bullet f^\dagger \bullet \mathcal{A}_{\trev_A^\dagger},
    \end{equation}
    where $\mathcal{A}_{\trev_B}:B \multimap B$ and $\mathcal{A}_{\trev_A^\dagger}:A \multimap A$ are supermaps defined by
    \begin{equation}
        \mathcal{A}_\trev (\chan) = \adj_{\trev^\dagger} \circ \chan \circ \adj_{\trev}.
    \end{equation}
\end{lemma}
\begin{proof}
    Local tomography in quantum theory uniquely identifies the CP map $\mathcal{M}_{f_\trev}$ satisfying \eqref{time_reversal_superchannel} as the RPS of $\mathcal{M}_f$. By Theorem~\ref{thm:trs_channel}, this map is
    \begin{align}
        (\mathcal{M}_f)_\trev &= \adj_{\trev_A^\dagger \otimes \trev_B^\dagger} \circ \mathcal{M}_f^\dagger \circ \adj_{\trev_A \otimes \trev_B} \\
        &= \adj_{\trev_A^\dagger \otimes \trev_B^\dagger} \circ \mathcal{M}_{f^\dagger} \circ \adj_{\trev_A \otimes \trev_B} \\
        &= \mathcal{M}_{\mathcal{A}_{\trev_B} \bullet f^\dagger \bullet \mathcal{A}_{\trev_A^\dagger}}.
    \end{align}
    This implies \eqref{tr_superchannel}.
\end{proof}
The Choi-operator expression of the RPS of $f$ is given by
\begin{equation}
    \label{choi_time_reversal_superchannel} \choi[f_\trev] = \adj_{X^\dagger_\mathrm{in} \otimes Y^\dagger_\mathrm{out} \otimes X^\top_\mathrm{out} \otimes Y^\top_\mathrm{in}} \circ \swapio (\choi[f]),
\end{equation}
where the unitary operators $X$ and $Y$ specify the motion-reversal operators through $\trev_A = \conj X$ and $\trev_B = \conj Y$, and $\swapio$ denotes the swap between input and output systems.

A simple class of T-symmetric superchannels is constructed by using mutually-transposed unital channels for both the encoder and decoder, as depicted in Figure~\ref{fig:superchannels} (b). In this example we assume that the motion reversal is given by complex conjugation $\conj$.
The superchannel of type $A \multimap I$ that discards its input channel (depicted in Figure~\ref{fig:superchannels} (c)) is also T-symmetric for any motion-reversal operator, even though the encoder-decoder pair is not mutually-transposed.

As an immediate consequence of the definition, T-symmetric superchannels from the void system $I$ are precisely the T-symmetric channels:
\begin{equation}
    \trs(I,A) = \trs(A).
\end{equation}
More generally, T-symmetric superchannels in $\trs(A,B)$ correspond exactly to the free operations for the T-violation resource theory.
\begin{definition}
    A deterministic superchannel $f:A \multimap B$ is said to be T-free if, for every ancillary system $E = (\hil_E,\trev_E)$ and every $\chan \in \trs(A \otimes E)$, one has $f \otimes i_E (\chan) \in \trs(B \otimes E)$.
\end{definition}
\begin{theorem}
    A deterministic superchannel is T-symmetric if and only if it is T-free.
\end{theorem}
\begin{proof}
    If $f:A \multimap B$ is T-free, $f \otimes i_A (S_A)$ is in $\trs(A \otimes B)$ for the swap channel $S_A$ on $\hil_A \otimes \hil_A$, since $S_A \in \trs(A \otimes A)$ when the motion-reversal operator is given by $\trev_A \otimes \trev_A$. The T-symmetry of $f$ follows since $\mathcal{M}_f = f \otimes i_A (S_A)$.

    Conversely, let us assume $f \in \trs(A,B)$ and $\chan \in \trs(A \otimes E)$. Their composition satisfies
    \begin{align}
        f \otimes i_E (\chan) &= f_\trev \otimes i_E (\chan_\trev) \\
        &= (\mathcal{A}_{\trev_B} \bullet f^\dagger \bullet \mathcal{A}_{\trev_A^\dagger}) \otimes i_E (\mathcal{A}_{\trev_A} \otimes \mathcal{A}_{\trev_E} (\chan^\dagger)) \\
        &= \mathcal{A}_{\trev_B} \otimes \mathcal{A}_{\trev_E} (f^\dagger \otimes i_E (\chan^\dagger)) \\
        &= \mathcal{A}_{\trev_B} \otimes \mathcal{A}_{\trev_E} ((f \otimes i_E (\chan))^\dagger) \\
        &= (f \otimes i_E (\chan))_\trev.
    \end{align}
    Since this holds for any ancillary system $E$ and any $\chan \in \trs (A \otimes E)$, $f$ is T-free.
\end{proof}

It is straightforward to verify that identity superchannels are T-symmetric and that T-symmetric superchannels are closed under sequential composition $\bullet$. They therefore define a dynamical resource theory \cite{ji_entropic_2023}. The superchannels from the void system $I$ are exactly the T-symmetric channels, and the general T-symmetric superchannels form the complete set of completely T-symmetry-preserving superchannels. We therefore identify this as the resource theory of \textit{T-violation}.

\subsection{Relation between the two classes of superchannels}
We have defined two classes of superchannels: unital and T-symmetric. These classes respectively define dynamical resource theories of nonunitality and T-violation \cite{ji_entropic_2023}. However, formulating T-violation within the class of unital channels---which we call \textit{kinematic T-violation}---requires a further comparison between the two classes.

In general, for two sets of free states or channels $F_1$ and $F_2$ satisfying
\begin{equation}
    F_1 \subset F_2,
\end{equation}
it appears natural to quantify deviations from $F_1$ within $F_2$ as a resource. A potential difficulty is that a free operation $O_1$ for $F_1$ may satisfy
\begin{equation}
    \label{stick_out}   O_1(F_2) \nsubseteq F_2.
\end{equation}
In this case, a resource initially in $F_2$ can be mapped outside $F_2$ by a free operation. Consequently, $F_2$ does not provide a closed domain for the restricted resource theory. This difficulty does not arise when the ambient domain includes all states or channels under consideration.

Here, we give an example of free sets realizing \eqref{stick_out}.
In the (static) resource theory of bipartite entanglement~\cite{Horodecki2009}, the class of states with positive partial transpose (PPT) contains that of separable states.
This inclusion is strict in all finite dimensions other than the two-qubit and qubit–qutrit cases~\cite{HORODECKI1996,HORODECKI1997}.
Although local operations and classical communication (LOCC)---the canonical free operations in entanglement theory---preserve both classes~\cite{Horodecki1998}, certain operations beyond LOCC can transform a PPT state into an entangled state with negative partial transpose (NPT) while preserving the set of separable states~\cite{Chitambar_beyondLOCC_2020}.
Thus, whether non-PPTness admits a resource-theoretic description under the same free operations as in non-separability depends on the choice of those operations.

Fortunately, in addition to
\begin{equation}
    \trs(A) \subset \unital(A) \qquad (\forall A),
\end{equation}
the unital and T-symmetric superchannels are related by
\begin{equation}
    \label{superchannel_subclass}   \trs(A,B) \subset \unital(A,B) \qquad (\forall A,B).
\end{equation}
Indeed, when $f \in \trs(A,B)$, Eq.~\eqref{tr_superchannel} implies that $f^\dagger$ is a deterministic superchannel. For any $\chan \in \unital(A \otimes E)$, the adjoint
\begin{equation}
    (f \otimes i_E (\chan))^\dagger = f^\dagger \otimes i_E (\chan^\dagger),
\end{equation}
is a valid channel since $\chan^\dagger$ is a channel and $f^\dagger$ is a deterministic superchannel.
By \eqref{superchannel_subclass}, the difficulty in \eqref{stick_out} does not arise for kinematic T-violation. This permits the construction of the nested resource theories summarized in Table~\ref{tab:nest}.

The inclusion \eqref{superchannel_subclass} lifts to a subcategory relation between the categories of T-symmetric and unital superchannels. Our resource theories can thus be formalized in the language of ``partitioned process theory'' \cite{coecke_mathematical_nodate}. Such a theory consists of a symmetric monoidal category $\cat$ of admissible processes together with a symmetric monoidal subcategory $\cat_\mathrm{free}$ of free processes. This categorical pair provides a minimal structure within which convertibility and resource monotones can be formulated systematically. The technical details are presented in Appendix~\ref{sec:partition}.

\section{Decomposition of overall T-violation}\label{sec:decomposition}
In Sec.~\ref{subsec:trs_channels}, we identified two distinct contributions to T-violation in general quantum channels. The first is the intrinsic T-violation of unital channels, analyzed in detail in Sec.~\ref{subsec:trs_superchannels}; we refer to this contribution as \emph{kinematic T-violation}. The second is nonunitality analyzed in Sec.~\ref{subsec:unital}, which is present only for nonunital channels and is independent of the choice of motion-reversal operator.
We expect these distinct contributions to satisfy the quantitative relation \eqref{concept}.
In this section, we formulate this relation rigorously as a quantitative decomposition of the corresponding resource monotones.

\subsection{Resource monotones for T-violation and nonunitality}
Resource monotones are non-negative functions of quantum channels that vanish precisely on free channels and are non-increasing under the corresponding free superchannels. For general T-violation, the free channels are the T-symmetric channels and the free transformations are T-symmetric superchannels. For nonunitality, the free channels and transformations are the unital channels and unital superchannels, respectively.

A general method for constructing monotones for dynamical resources was given in \cite{liu_operational_2020}. The construction is based on channel distinguishability. Let $\free$ denote the set of free channels in the resource theory under consideration, and let $d$ be any distance measure on quantum states. Then
\begin{equation}
    \label{monotone_general}    \inf_{\chan_\mathrm{free} \in \free} \sup_\rho d(\chan \otimes \iden (\rho), \chan_\mathrm{free} \otimes \iden (\rho)),
\end{equation}
where $\rho$ ranges over states on a sufficiently enlarged system, is a resource monotone \cite{liu_operational_2020}.

Although \eqref{monotone_general} can be applied directly to T-violation, we instead use a construction based on Choi states. This approach is central to our integration theorem and also improves the computability and resolving power of monotones for kinematic T-violation, as discussed in Sec.~\ref{sec:kinematic}. Define the Choi \emph{state} of a channel $\chan:\B(\hil) \rightarrow \B(\hil)$ by
\begin{equation}
    \label{cstate}  \cstate[\chan] := \chan \otimes \iden \left( \frac{1}{d}\sum_{i,j=1}^d \ket[i,i] \bra[j,j] \right),
\end{equation} 
where $d= \dim \hil$. Thus, Choi states are normalized Choi operators. The key observation is that unital superchannels act as ordinary quantum channels on these Choi states.
\begin{lemma}\label{lem:choistate}
    For any unital superchannel $f \in \unital(A,B)$, the map $\widetilde{f}:\B(\hil_A \otimes \hil_A) \rightarrow \B(\hil_B \otimes \hil_B)$ defined by
    \begin{equation}
        \widetilde{f}:\cstate[\chan] \mapsto \cstate[f(\chan)],
    \end{equation}
    is linear, completely positive, and trace-preserving.
\end{lemma}
\noindent The proof of this lemma is presented in Appendix~\ref{sec:proof_choi_CPTP}.
A general deterministic superchannel need not act as a channel on Choi states. A simple example is the ``state-preparation'' superchannel, which feeds a fixed state $\rho$ into the input channel.

Lemma~\ref{lem:choistate} directly yields a family of resource monotones:
\begin{theorem}\label{thm:monotones}
    Let $d$ be any distance function or any divergence between pairs of quantum states that is non-increasing under quantum channels. Then the functions
    \begin{align}
        \label{monotone_trs}    \Delta^\trs_d (\chan) &:= \inf_{\chan_\trs \in \trs(A)} d(\cstate[\chan],\cstate[\chan_\trs]),\\
        \label{monotone_unital} \Delta^\unital_d (\chan) &:= \inf_{\chan_\unital \in \unital(A)} d(\cstate[\chan],\cstate[\chan_\unital]),
    \end{align}
    are resource monotones for T-violation and nonunitality, respectively.
\end{theorem}
\noindent When $\chan$ is unital, we refer to its T-violation specifically as \emph{kinematic T-violation} (see Sec.~\ref{sec:kinematic} for details).
In general, Choi-state distances need not define channel-resource monotones; their monotonicity here relies essentially on the channel representation established in Lemma~\ref{lem:choistate}.

Our construction differs from the Choi-defined resource theories studied in \cite{zanoni_choi-defined_2024}. A channel resource theory is Choi-defined when the Choi states of its free channels are precisely free states of an underlying state resource theory. Nonunitality is not Choi-defined in this sense, because Choi states of unital channels need not be maximally mixed. Likewise, the T-violation considered here has no underlying state-resource counterpart. Lemma~\ref{lem:choistate} therefore provides a distinct role for Choi states: they furnish a state-space representation of the relevant free superchannels even when the channel resource itself is not Choi-defined.


In this initial survey on T-violation, we do not further explore the physical interpretation of nonunitality as a resource, beyond noting its thermodynamic significance and its independence from motion reversal. Within a top-down framework \cite{ji_entropic_2023}, nonunitality can be viewed as the channel analogue of state nonuniformity \cite{gour_resource_2015}; it is thus expected to both generate and be generated by states that are not maximally mixed. In statistical mechanics, unitality is known to play the role of the microreversibility condition \cite{albash_fluctuation_2013} in establishing fluctuation relations for general quantum channels \cite{manzano_nonequilibrium_2015,rastegin_non-equilibrium_2013}. Beyond unital open-system dynamics, nonunitality necessitates a correction to the Jarzynski equality \cite{rastegin_jarzynski_2014} and can induce entropy production in the environment \cite{aurell_time_2015}. A rigorous operational characterization of nonunitality as a resource, including its exact role in statistical mechanics, lies beyond the scope of the present work.

\subsection{Pythagorean theorem for T-violation}
We apply the generalized Pythagorean theorem for quantum statistical manifolds \cite{matsuda_information_2022} to relate kinematic T-violation and nonunitality. As described below, this application extends to other nested resources.

Let $S(\rho):=-\trace[\rho\log\rho]$ denote the von Neumann entropy and let $S(\rho||\sigma):=-S(\rho)-\trace[\rho\log\sigma]$ denote the quantum relative entropy.

We review several concepts from information geometry that are adapted to the quantum state space (see \cite{amari_methods_2007,matsuda_information_2022}).
Let $\density$ denote the $(\dim \hil)^2-1$ dimensional manifold of positive density matrices on $\hil$.
This manifold is embedded in $\mathbb{R}^{(\dim \hil)^2}$ by choosing $(\dim \hil)^2$ real and imaginary parts of matrix elements, subject to $\trace \rho =1$ and $\rho > 0$, for the coordinates. We assume the Bogoliubov--Kubo--Mori metric \cite{petz1993bogoliubov} on this manifold in the following argument.

A submanifold $A$ of $\density$ is \textit{m}-totally geodesic connected (\textit{m}-TGC) when it is closed under convex mixtures. When $A$ is $m$-TGC, for an arbitrary state $\rho \in \density$, the minimum
\begin{equation}
    \min_{\sigma \in A} S(\sigma || \rho)
\end{equation}
is attained at a unique point in $A$, called the \textit{e}-projection of $\rho$ onto $A$. We denote the \textit{e}-projection by
\begin{equation}
    \sigma_A^m (\rho),
\end{equation}
when it exists.
A submanifold $A$ of $\density$ is \textit{e}-totally geodesic connected (\textit{e}-TGC) when it is closed under the following ``mixture in \textit{e}-representation'' of an arbitrary pair $\rho_1, \rho_2 \in A$:
\begin{equation}
    C(t)_{\rho_1,\rho_2} := \frac{\exp \left( t \log \rho_1 + (1-t) \log \rho_2 \right)}{\trace\left[ \exp \left( t \log \rho_1 + (1-t) \log \rho_2 \right) \right]}.
\end{equation}
When $A$ is $e$-TGC, for an arbitrary state $\rho \in \density$, the minimum
\begin{equation}
    \min_{\sigma \in A} S(\rho || \sigma)
\end{equation}
is attained at a unique point in $A$, called the \textit{m}-projection of $\rho$ onto $A$. We denote the \textit{m}-projection by
\begin{equation}
    \sigma_A^e (\rho),
\end{equation}
when it exists.

In our example, the convex sets of Choi states of unital channels and T-symmetric channels both form \textit{m}-TGC submanifolds.\footnote{They are submanifolds since $\density$ is a convex manifold with an interior point $\id/\dim \hil$, and the affine constraints defining the subsets both contains the point as a solution.}
Exceptionally, they are \textit{e}-TGC when the dimension of the input (and hence the output) system is $2$ (see Appendix~\ref{sec:flat} for details).

We have the following two kinds of generalized Pythagorean theorems adapted to nested resource theories.
\begin{lemma}\label{lem:pythagorean}
    If $A_1$ and $A_2$ are \textit{m}- (\textit{e}-)TGC submanifolds of $\density$ satisfying $A_1 \subset A_2$, then, for every $\rho\in\density$,
    \begin{equation}
        \label{pythagorean1}    \sigma_{A_1}^{m(e)}(\rho) = \sigma_{A_1}^{m(e)} \left( \sigma_{A_2}^{m(e)} (\rho) \right).
    \end{equation}
    When $A_1$ and $A_2$ are \textit{m}-TGC,
    \begin{equation}
        \label{pythagorean2}    S(\sigma_{A_1}^m (\rho) || \rho) = S(\sigma_{A_1}^m (\rho) || \sigma_{A_2}^m (\rho)) + S(\sigma_{A_2}^m(\rho) || \rho ).
    \end{equation}
    When $A_1$ and $A_2$ are \textit{e}-TGC,
    \begin{equation}
        \label{pythagorean_e}    S( \rho || \sigma_{A_1}^e (\rho)) = S( \rho || \sigma_{A_2}^e(\rho) ) + S(\sigma_{A_2}^e (\rho) || \sigma_{A_1}^e (\rho)).
    \end{equation}
\end{lemma}
\begin{proof}
    We prove the lemma for \textit{m}-TGC submanifolds and \textit{e}-projections. The proof for \textit{e}-TGC case is obtained by exchanging \textit{m} and \textit{e}, and reversing the arguments in relative entropies.

    Since $\sigma_{A_2}^m(\rho)$ is the \textit{e}-projection of $\rho$ onto $A_2$, while $\sigma_{A_1}^m(\rho)$ and $\sigma_{A_1}^m( \sigma_{A_2}(\rho))$ both belong to $A_1 \subset A_2$, the generalized Pythagorean theorem (see \cite{amari_methods_2007} for general theory, and Lemmas A.1 and A.2 of \cite{matsuda_information_2022} for this setting) give
    \begin{align}
        \label{pythagorean3}    & S ( \sigma_{A_1}^m(\rho) || \rho ) = S(\sigma_{A_1}^m(\rho) || \sigma_{A_2}^m(\rho)) + S(\sigma_{A_2}^m(\rho) || \rho),\\
        \nonumber & S(\sigma_{A_1}^m( \sigma_{A_2}(\rho)) || \rho ) \\
        \label{pythagorean4}    &= S(\sigma_{A_1}^m(\sigma_{A_2}^m(\rho)) || \sigma_{A_2}^m(\rho)) + S( \sigma_{A_2}^m(\rho) || \rho).
    \end{align}
    The first equality is precisely \eqref{pythagorean2}.
    Because $\sigma_{A_1}^m(\rho)$ minimizes $S(\sigma || \rho)$ over $A_1$, we have $S(\sigma_{A_1}^m(\rho) || \rho) \leq S(\sigma_{A_1}^m( \sigma_{A_2}^m(\rho)) || \rho)$. Conversely, because $\sigma_{A_1}^m( \sigma_{A_2}^m(\rho))$ minimizes $S(\sigma || \sigma_{A_2}^m(\rho))$ over the same set, $S(\sigma_{A_1}^m(\rho) || \sigma_{A_2}^m(\rho)) \geq S(\sigma_{A_1}^m( \sigma_{A_2}^m(\rho))||\sigma_{A_2}^m(\rho))$. Equations~\eqref{pythagorean3} and \eqref{pythagorean4} force both inequalities to be equalities. Uniqueness of the minimizer therefore implies $\sigma_{A_1}^m(\rho)=\sigma_{A_1}^m( \sigma_{A_2}^m(\rho))$, proving \eqref{pythagorean1}.
\end{proof}
When free resources are provided by nested TGC submanifolds, Lemma~\ref{lem:pythagorean} constitutes a Pythagorean theorem for nested resources. In particular, Eqs.~\eqref{pythagorean2} and \eqref{pythagorean_e} follows from the orthogonality between the geodesic from $\rho$ to $\sigma_{A_2}(\rho)$ and that from $\sigma_{A_2}(\rho)$ to $\sigma_{A_1}(\rho)$. The divergence associated with the metric (Bogoliubov--Kubo--Mori metric) is the quantum relative entropy. Figure~\ref{fig:geometry} illustrates the orthogonality relation used in Lemma~\ref{lem:pythagorean}.
\begin{figure}[tbp]
    \includegraphics[width=.45\textwidth]{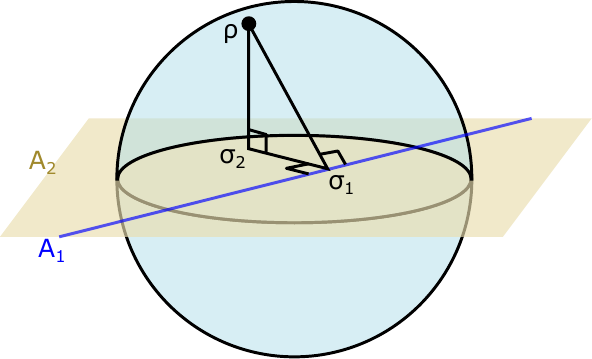}
    \caption{\label{fig:geometry} Schematic relation among the TGC submanifolds and minimizers in Lemma~\ref{lem:pythagorean}. The ball represents the state space $\density(\hil)$. The geometry is deliberately simplified to depict the generalized Pythagorean theorem in familiar Euclidean terms.
    }
\end{figure}

The relation between coherence and imaginarity provides a simple illustration of Lemma~\ref{lem:pythagorean}.
Coherence, imaginarity, and residual coherence in real states satisfy a triangle inequality across a wide range of resource monotones \cite{guo_imaginarity_2026}, whereas the relative entropy monotone turns this into an equality.
For the nested \textit{e}-TGC submanifolds of real ($A_2$) and incoherent ($A_1$) density matrices, Eq.~\eqref{pythagorean1} specializes to
\begin{equation}
    \label{diag_real}    \diag(\rho) = \diag \left( \real(\rho) \right).
\end{equation}
The closest incoherent and real states to $\rho$ (with respective to the specific relative entropy monotones) are, respectively, $\diag(\rho):=\rho_{11} \oplus \cdots \oplus \rho_{nn}$ and $\real(\rho):=(\rho+\rho^\top)/2$.
Equation~\eqref{diag_real} also yields the corresponding decomposition \eqref{pythagorean_e}. The lemma shows that this Pythagorean structure persists quite generally, even when the closest free states do not admit closed-form expressions.

We can now state one of the main results of this article, which gives a rigorous formulation of Eq.~\eqref{concept}. Define the relative entropy monotones of T-violation by
\begin{align}
    \Delta_m^\trs (\chan) &:= \min_{\chan_\trs \in \trs(A)} S(\cstate[\chan_\trs] || \cstate[\chan]),\\
    \label{monotone_choi_trs}   \Delta_e^\trs (\chan) &:= \min_{\chan_\trs \in \trs(A)} S(\cstate[\chan] || \cstate[\chan_\trs]),
\end{align}
and those of nonunitality by
\begin{align}
    \Delta_m^\unital (\chan) &:= \min_{\chan_\unital \in \unital(A)} S(\cstate[\chan_\unital] || \cstate[\chan]),\\
    \label{monotone_choi_unital}   \Delta_e^\unital (\chan) &:= \min_{\chan_\unital \in \unital(A)} S(\cstate[\chan] || \cstate[\chan_\unital]),
\end{align}
according to Theorem~\ref{thm:monotones}.
Although subscripts ``\textit{m}'' and ``\textit{e}'' come from the relevant structures of free states, the \textit{e}-type monotones \eqref{monotone_choi_trs} and \eqref{monotone_choi_unital} remain well-defined even when the free-Choi-state spaces are not \textit{e}-TGC.
\begin{theorem}\label{thm:pythagorean_trs}
    For any channel $\chan$ on a finite-dimensional system $A$ with a strictly positive Choi state $\cstate[\chan] > 0$,
    \begin{equation}
        \label{pythagorean_trs_m} \Delta^\trs_m (\chan) = \Delta^\unital_m (\chan) + \Delta^\trs_m \left( \argmin_{\chan' \in \unital(A)} S(\cstate[\chan'] || \cstate[\chan]) \right).
    \end{equation}
    For any channel $\chan$ on a $2$-dimensional system $A$ with a strictly positive Choi state $\cstate[\chan] >0$,
    \begin{equation}
        \label{pythagorean_trs_e} \Delta^\trs_e (\chan) = \Delta^\unital_e (\chan) + \Delta^\trs_e \left( \argmin_{\chan' \in \unital(A)} S( \cstate[\chan] || \cstate[\chan']) \right).
    \end{equation}
\end{theorem}
\noindent These identities follow directly from Lemma~\ref{lem:pythagorean}. See Appendix~\ref{sec:flat} for the proof that Choi-state spaces for $2$-dimensional unital and T-symmetric channels are both \textit{e}-TGC. The relative-entropy projection onto the smaller free set therefore decomposes into the projection onto the unital set followed by the residual projection onto the T-symmetric set.

At the moment, we do not have an extension of the decomposition theorem for channels with a deficient Choi-rank. 

One may ask whether relative entropy is essential for establishing the decomposition. We do not completely deny other possibilities, but the generalized Pythagorean theorem, and hence Lemma~\ref{lem:pythagorean}, relies strongly on the torsion-free property of the Bogoliubov--Kubo--Mori metric. If another metric is used instead, the generalized Pythagorean identity requires a correction term arising from torsion \cite{henmi2018statistical}. This obstructs a straightforward generalization of the Pythagorean theorem for nested resources.

\subsection{Relative entropies of T-violation: general properties}
Motivated by the central role of relative entropy monotones $\Delta_{m \backslash e}^{\trs \backslash \unital}$ in the main decomposition theorem, we explore their general properties. The \textit{e}-type kinematic T-violation (i.e., $\Delta_e^\trs$ applied to unital channels) will be further analyzed with concrete examples in Section~\ref{sec:kinematic}.

A benefit of \textit{m}-type relative entropy monotones is their potential to be strongly additive. This was observed for entanglement in \cite{eisert_remarks_2003}, whose proof we adapt to T-violation.
\begin{lemma}
    For any pair of channels $\chan_A$ on system $A$ and $\chan_B$ on system $B$,
    \begin{align}
        \sigma_{\trs \backslash \unital}^m (\chan_A \otimes \chan_B) &= \sigma_{\trs \backslash \unital}^m (\chan_A) \otimes \sigma_{\trs \backslash \unital}^m (\chan_B),\\
        \Delta^{\trs \backslash \unital}_m (\chan_A \otimes \chan_B) &= \Delta^{\trs \backslash \unital}_m (\chan_A) + \Delta^{\trs \backslash \unital}_m (\chan_B).
    \end{align}
\end{lemma}
\begin{proof}
    Generally, for any bipartite quantum state $\tau$ and product state $\rho_A \otimes \rho_B$, we have \cite{eisert_remarks_2003}
    \begin{align}
        \label{strong_additivity1}  S(\tau || \rho_A \otimes \rho_B) &\geq S(\trace_B[\tau] \otimes \trace_A[\tau] || \rho_A \otimes \rho_B) \\
        \label{strong_additivity2}  &= S(\trace_B[\tau] || \rho_A ) + S(\trace_A[\tau] || \rho_B).
    \end{align}
    Now we take $\tau = \sigma^m_{\trs \backslash \unital}(\chan_A \otimes \chan_B)$, $\rho_A = \cstate[\chan_A]$, and $\rho_B = \cstate[\chan_B]$. Taking the partial trace of the Choi state over $A$, for instance, corresponds to applying a T-symmetric (or unital) superchannel that inputs the completely mixed state on $\hil_A^\mathrm{in}$ and traces out $\hil_A^\mathrm{out}$ (see Figure~\ref{fig:superchannels} (c)). Therefore, $\trace_A \left[ \sigma^m_{\trs \backslash \unital}(\chan_A \otimes \chan_B) \right]$ is a valid Choi state of a T-symmetric$\backslash$unital channel, and the same holds for $\trace_B \left[ \sigma^m_{\trs \backslash \unital}(\chan_A \otimes \chan_B) \right]$. Since
    \begin{equation}
        \trace_B \left[ \sigma^m_{\trs \backslash \unital}(\chan_A \otimes \chan_B) \right] \otimes \trace_A \left[ \sigma^m_{\trs \backslash \unital}(\chan_A \otimes \chan_B) \right]
    \end{equation}
    is a valid Choi state of a T-symmetric$\backslash$unital channel, \eqref{strong_additivity1} implies that it is the \textit{e}-projection of (the Choi state of) $\chan_A \otimes \chan_B$. The lemma follows from \eqref{strong_additivity2} and the uniqueness of \textit{e}-projections.
\end{proof}
It will be evident from the results of Section~\ref{subsec:monotones} that strong additivity does not hold for \textit{e}-type relative entropy monotones.

While \textit{e}-type relative entropy monotones have been proposed for various convex resource theories, including entanglement, coherence, and imaginarity, \textit{m}-type ones have been actively studied recently due to their central role in the quantum Sanov theorem \cite{hayashi_entanglement_2025,lami_asymptotic_2024,lami_generalised_2025}.
A difficulty in handling \textit{m}-type monotones, already mentioned in \cite{eisert_remarks_2003}, is that they easily diverge when the target state has low rank. Indeed, the relative entropy $S(\sigma || \rho)$ becomes $+ \infty$ whenever the support of $\sigma$ is not contained in the support of $\rho$. Therefore, \textit{m}-type monotones are sometimes too coarse to evaluate the strength of convex resources. For this reason, we focus on \textit{e}-type monotones when exploring the convertibility of kinematic T-violation resources in Section~\ref{sec:kinematic}.

We close this section by a remark on the underlying minimizations involved in our relative-entropy measures.
These minimizations are generally difficult to perform analytically.
In our setting, Theorem~\ref{thm:twirl} determines the closest T-symmetric channel to a given unital channel in the second argument, leading to a closed-form expression for $\Delta^\trs_e$ on unital channels.
Finding the unital channel closest to a given channel is a central problem in operator scaling \cite{gurvits_classical_2003,idel_review_2016,allen-zhu_operator_2018,matsuda_information_2022} and in the quantum Schr\"odinger bridge problem \cite{georgiou_positive_2015}. For classical stochastic matrices, the closest doubly stochastic matrix (in the first argument of relative entropy) can be found iteratively using Sinkhorn's algorithm \cite{sinkhorn_relationship_1964,sinkhorn_concerning_1967,csiszar_i-divergence_1975}. This procedure generalizes to operator scaling for quantum channels \cite{gurvits_classical_2003}. Although operator scaling shares the geometric interpretation of Sinkhorn's algorithm, it does not necessarily converge to the closest unital channel in Choi relative entropy \cite{matsuda_information_2022}. The quantum Schr\"odinger bridge approach likewise constructs a family of unital channels from a given nonunital channel \cite{georgiou_positive_2015}, but the distance between the original channel and its ``bridge'' remains unexplored.

\section{Kinematic T-violation}\label{sec:kinematic}
In this section, we study the convertibility of unital channels under T-symmetric superchannels and introduce monotone functions for this conversion preorder. This is equivalent to studying the resource theory of kinematic T-violation.

We identify a channel with maximal kinematic T-violation in the smallest nontrivial, two-dimensional system. This channel serves as a golden unit: a universal resource that can be converted into any unital channel by a T-symmetric superchannel. Its existence implies that kinematic T-violation is not an extensive resource that simply increases with system size. Indeed, we find that kinematic T-violation is neither quantity-like nor quality-like.

We also identify a classical golden unit: a unital classical channel (i.e., a doubly stochastic map) that can be converted into any unital channel. Thus, under our definition of T-symmetry, T-violation is not an exclusively quantum phenomenon.

Resource monotones provide necessary conditions for convertibility. We discuss both the standard construction of dynamical resource monotones \cite{liu_operational_2020} and alternative constructions \eqref{monotone_trs} tailored to the present theory.

We do not, however, obtain necessary and sufficient criteria for convertibility, even after restricting to pure resources, i.e., unitary channels. Such criteria are known in particular resource theories, including entanglement, coherence, and imaginarity. A complete characterization of convertibility for T-violation remains an important direction for future work.


\subsection{Interplay between different motion-reversal operators}
The construction of free superchannels allows us to discuss resource conversion not only within a fixed system, but also between systems equipped with different motion-reversal operators. This makes it possible to compare resources defined relative to distinct notions of motion reversal. We begin with isomorphic conversions between systems of the same dimension but with different motion-reversal operators.

Two systems have the same resource-convertibility structure when they are related by a free isomorphism. Indeed, suppose that $i:A \multimap B$ is a free operation and that its inverse $i^{-1}:B \multimap A$ exists and is also free, with inverse understood in the sense that $i \circ i^{-1} = \iden$. Then every resource conversion within $A$ has an isomorphic counterpart in $B$. More explicitly, if a free operation $f_A$ on $A$ satisfies $f_A(a_1)=a_2$, then $f_B := i \circ f_A \circ i^{-1}$ is a free operation on $B$ satisfying $f_B(b_1)=b_2$, where $b_1=i(a_1)$ and $b_2=i(a_2)$. Conversely, one can reconstruct $f_A$ from $f_B$. The superchannels $i$, $i^{-1}$, $f_A$ and $f_B$ are related by the commutative diagram:
\[
\begin{tikzcd}
  A \ar[r,yshift=0.5ex, "i"] \arrow[d, "f_A"'] & B \ar[d, "f_B"] \ar[l,yshift=-0.5ex, "i^{-1}"] \\
  A \ar[r,yshift=0.5ex, "i"]  & B.  \ar[l,yshift=-0.5ex, "i^{-1}"]
\end{tikzcd}
\]
The two systems are therefore equivalent as resource theories. This is the same logic that permits one to fix the computational basis in resource theories of coherence and imaginarity.

We now characterize such isomorphisms for systems of equal dimension equipped with different motion-reversal operators.
\begin{theorem}\label{thm:isomorphism}
    Let $A = (\hil_A, \conj X)$ and $B = (\hil_B, \conj Y)$ be systems having the same dimension $\dim \hil_A = \dim \hil_B$.
    There exists an invertible T-symmetric superchannel $f \in \trs(A,B)$ if and only if $X^\dagger X^\top$ and $Y^\dagger Y^\top$ share the same spectrum up to a global phase. When it exists, the inverse $f^{-1}$ belongs to $\trs(B,A)$.
\end{theorem}
\begin{proof}
    A deterministic superchannel between systems of the same dimension is invertible if and only if it is implemented by a unitary encoder--decoder pair with no memory system. Let $E$ and $D$ denote the unitary transformations implementing the encoder and decoder, respectively. The Choi operator of the superchannel is then
    \begin{equation}
        \label{choi_isomorphism}    \choi[f] = \adj_{E_A^\mathrm{in}} (\Psi_{B_\mathrm{in}:B_\mathrm{in}}) \otimes \adj_{D_B^\mathrm{out}} (\Psi_{A_\mathrm{out}:A_\mathrm{out}}),
    \end{equation}
    where $\Psi$ denotes the unnormalized maximally entangled operator on the indicated spaces. The unital superchannel $f$ is T-symmetric if and only if
    \begin{equation}
        \label{choi_symmetric}  \choi[f] = \adj_{X^\dagger_\mathrm{in} \otimes Y^\dagger_\mathrm{out} \otimes X^\top_\mathrm{out} \otimes Y^\top_\mathrm{in}} \circ \swapio (\choi[f]),
    \end{equation}
    from \eqref{choi_time_reversal_superchannel}, where $\trev_A = \conj X$ and $\trev_B = \conj Y$. By substituting \eqref{choi_isomorphism} into \eqref{choi_symmetric}, we obtain the condition
    \begin{equation}
        \label{condition1}    D \propto Y^\dagger E^\top X, \quad E \propto X^\dagger D^\top Y.
    \end{equation}
    Eliminating $E$ gives the equivalent condition
    \begin{equation}
        \label{condition2}  D X^\dagger X^\top D^\dagger \propto Y^\dagger Y^\top.
    \end{equation}
    When this condition holds, one may take $E = X^\dagger D Y$ to complete the construction of an invertible superchannel in $\trs(A,B)$. Condition \eqref{condition2} is precisely the statement that $X^\dagger X^\top$ and $Y^\dagger Y^\top$ have the same spectrum up to a global phase.

    The inverse of $f$ is implemented by the encoder $E' = E^\dagger$ and decoder $D' = D^\dagger$. For this inverse to be T-symmetric, these unitaries must satisfy
    \begin{equation}
        \label{condition3}    D' \propto X^\dagger {E'}^\top Y, \quad E' \propto Y^\dagger {D'}^\top X,
    \end{equation}
    analogously to \eqref{condition1}. This condition is automatically satisfied when $E' = E^\dagger$, $D' = D^\dagger$, and \eqref{condition1} holds.
\end{proof}

An antiunitary operator $\trev$ is called a conjugation when $\trev^\dagger = \trev$, and a skew-conjugation when $\trev^\dagger = - \trev$. If $\trev = \conj X$ is a conjugation or a skew-conjugation, then $X^\top = X$ or $X^\top = -X$, respectively, and hence $X^\dagger X^\top = \id$ or $X^\dagger X^\top = -\id$. Thus, systems whose motion-reversal operators are conjugations or skew-conjugations are related by isomorphic T-symmetric superchannels.

In general, the symmetry properties of antiunitary operators depend sensitively on the operator in question. The distinction between conjugations and skew-conjugations is especially important. Motion-reversal operators for bosons and fermions are typically modeled by conjugations and skew-conjugations, respectively. These two cases differ, for instance, in that conjugations have eigenvectors and define real subspaces, whereas skew-conjugations have no eigenvectors. Their distinction in bipartite systems is discussed in \cite{miyazaki_non-locality_2024}.
Nevertheless, the convertibility structures for conjugations and skew-conjugations, and hence for bosonic and fermionic motion reversal, are identical.

Because isomorphic systems have identical convertibility structures, it suffices to analyze the simplest representative. In what follows, we focus on systems whose motion-reversal operator is complex conjugation $\conj$ in the computational basis. Quantum channels are T-symmetric precisely when
\begin{equation}
    \chan = \chan^\top,
\end{equation}
in this convention. By Thm.~\ref{thm:isomorphism}, results obtained for this representative apply equally to all systems with conjugation or skew-conjugation motion-reversal operators.

\subsection{Resource monotones}\label{subsec:monotones}
We now construct functions that quantify kinematic T-violation. This includes the Choi-defined T-violation monotone \eqref{monotone_trs} applied to unital channels. Such monotones provide necessary conditions for convertibility: if $m(\chan)<m(\chan')$, then no T-symmetric superchannel can convert $\chan$ into $\chan'$.

\subsubsection{Monotones based on Choi states}
Monotones based on Choi states are non-increasing under unital superchannels, although not necessarily under all deterministic superchannels. Lemma~\ref{lem:choistate} implies that distances between Choi states are non-increasing under T-symmetric superchannels. In addition to \eqref{monotone_trs}, the following are monotones for any state-distance measure or divergence $d$:
\begin{align}
    \label{twirling_monotone}   \Delta_d^\mathrm{twirl} (\chan) &:= d(\cstate[\chan], \cstate[\chan_\mathrm{twirl}]),\\
    \label{flipped_monotone}    \Delta_d^\trev (\chan) &:= d(\cstate[\chan], \cstate[\chan_{\trev_A}]),
\end{align}
where the twirled channel $\chan_\mathrm{twirl}$ is defined by
\begin{equation}
    \label{twirling}    \chan_\mathrm{twirl} := \frac{\chan + \chan_{\trev_A}}{2}.
\end{equation}

Although the twirling monotone \eqref{twirling_monotone} is a monotone for general T-violation, it is particularly meaningful when $\chan$ is unital and hence $\chan_\mathrm{twirl}$ is a T-symmetric channel. For certain choices of $d$, twirling yields the closest Choi state of a T-symmetric channel in \eqref{monotone_trs}:
\begin{theorem}\label{thm:twirl}
    If $d$ is the trace distance, the Frobenius distance, or the quantum relative entropy (\textit{e}-type), then, for any unital channel $\chan_\unital$, its twirled channel $\chan_\mathrm{twirl}$ defined by \eqref{twirling} minimizes $d(\cstate[\chan_\unital], \cstate[\chan_\trs])$ over all T-symmetric channels $\chan_\trs$.
\end{theorem}
This reduction applies to a broad class of relative-entropy monotones quantifying symmetry breaking \cite{gour_measuring_2009}. The proofs for the trace distance and Frobenius norm are analogous to those in state resource theories such as imaginarity, but are included in Appendix~\ref{sec:twirl} for completeness. For these distance functions, \eqref{monotone_trs} coincides with \eqref{twirling_monotone}.

An important example of the flip monotone \eqref{flipped_monotone} is the fidelity between Choi states. This fidelity is comparatively easy to compute and provides useful insight into convertibility. For unitary channels, it is also directly related to the relative-entropy monotone.


\subsubsection{Example: relative entropy and fidelity of kinematic T-violation}
For a unital channel $\chan$, the relative-entropy monotone reduces to the twirled relative entropy (Theorem~\ref{thm:twirl} and Proposition~2 of \cite{gour_measuring_2009}):
\begin{align}
    \Delta^\trs_e (\chan) &= \Delta^\mathrm{twirl}_e (\chan) = S \left( \cstate[\chan] || \cstate[\chan_\mathrm{twirl}] \right) \\
    \label{entropy_difference}  &= S(\cstate[\chan_\mathrm{twirl}]) - S(\cstate[\chan]) \quad (\chan \in \unital).
\end{align}
The monotone is also convex:
\begin{equation}
    \label{convexity}   \Delta^\trs_e (p \chan + (1-p) \chan') \leq p \Delta^\trs_e (\chan)+ (1-p) \Delta^\trs_e (\chan').
\end{equation}
Its range is
\begin{equation}
    \label{range}   0 \leq \Delta^\trs_e (\chan) \leq 1,
\end{equation}
when $\chan$ and $\chan'$ are unital (see Appendix~\ref{sec:convexity} for the derivation).

For a unitary channel $\adj_U$, the relative entropy of T-violation reduces to
\begin{equation}
    \label{relative_entropy_fidelity}   \Delta^\trs_e (\adj_U) = h_2 \left( \frac{1+ \sqrt{\fidelity^\trev(\adj_U)}}{2} \right),
\end{equation}
where $h_2$ is the binary entropy $h_2(p) = -p \log p - (1-p) \log (1-p)$, and $\fidelity^\trev$ is the fidelity of T-violation
\begin{equation}
    \fidelity^\trev(\chan) := \left(\trace \sqrt{\sqrt{\cstate[\chan]} \cstate[\chan_\trev] \sqrt{\cstate[\chan]}}\right)^2.
\end{equation}
For a unitary channel $\chan = \adj_U$, this reduces to
\begin{equation}
    \fidelity^\trev(\adj_U) = \frac{1}{d^2} \left| \trace [ U_\trev^\dagger U] \right|^2,
\end{equation}
and, in the special case $\trev = \conj$, to
\begin{equation}
    \label{fidelity_unitary}    \fidelity^\conj(\adj_U) = \frac{1}{d^2} \left| \trace [ \overline{U} U] \right|^2.
\end{equation}
The fidelity takes values in $[0,1]$. Unlike the convention for resource monotones above, it equals $1$ on T-symmetric channels and is monotonically non-decreasing under T-symmetric superchannels.

As a concrete example, consider a single-qubit system $(\mathbb{C}^2, \conj)$. By unitary congruence (see \cite{horn_matrix_2012} $\S$ 4.4, or \cite{wigner_normal_1960} for this version), any $2 \times 2$ unitary operator can be uniquely decomposed as
\begin{align}
    \label{congruence}  U = V \left[ \begin{array}{cc}
        0 & e^{i \phi} \\
        e^{-i \phi} & 0
    \end{array}\right] V^\top, \quad (\phi \in [0,\pi/2])
\end{align}
for some unitary $V$, and the eigenvalues of $\overline{U}U$ are $e^{\pm 2 i \phi}$. Therefore, the fidelity is
\begin{equation}
    \label{fidelity_qubit}  \fidelity^\conj(\adj_U) = | \cos 2 \phi |^2.
\end{equation}
Examples of T-symmetric unitary channels, corresponding to $\phi = 0,\pi/2$, are
\begin{align}
    U = \left[ \begin{array}{cc}
        1 & 0 \\
        0 & 1
    \end{array}\right],~
    \left[ \begin{array}{cc}
        0 & 1 \\
        1 & 0
    \end{array}\right],~
    \left[ \begin{array}{cc}
        0 & -i \\
        i & 0
    \end{array}\right].
\end{align}
By contrast, unitary channels that maximally break T-symmetry correspond to $\phi=\pi/4$; examples are
\begin{align}
    \label{golden}  G = \left[ \begin{array}{cc}
        0 & e^{i \pi /4} \\
        e^{- i \pi /4} & 0
    \end{array}\right],~
    \frac{1}{\sqrt{2}} \left[ \begin{array}{cc}
        1 & 1 \\
        -1 & 1
    \end{array}\right].
\end{align}

It is worth emphasizing that the values $\fidelity^\conj=0$ and $\Delta^\trs_S = 1$ are already attainable in a single-qubit system. In Sec.~\ref{subsec:golden_unit}, we establish that $\adj_G$ is a golden unit, i.e., a resource from which any other unital channel can be generated by a T-symmetric superchannel.

A notable property of this fidelity is multiplicativity under parallel composition:
\begin{equation}
    \label{multiplicativity}    \fidelity^{\trev_A \otimes \trev_B} (\chan_A \otimes \chan_B) = \fidelity^{\trev_A} (\chan_A ) \times \fidelity^{\trev_B} (\chan_B),
\end{equation}
for any pair of channels $\chan_A \in \unital(A)$ and $\chan_B \in \unital(B)$. Here, the motion reversal of the composite system is taken to be $\trev_A \otimes \trev_B$. The identity follows directly from multiplicativity of the ordinary state fidelity. It implies that if a composite resource is maximal, in the sense that $\fidelity=0$, then at least one component resource must already be maximal. A tensor product of non-maximal resources can therefore never become maximal.

\subsubsection{Monotones based on channel distinguishability}\label{subsubsec:distinguishability}
Let us now follow the general method for constructing monotones for dynamical resources \cite{liu_operational_2020}, which is outlined by \eqref{monotone_general}. The construction is based on channel distinguishability.

The monotone \eqref{monotone_general} is difficult to compute because it involves nested optimizations.
Analogously to the Choi-state monotones, the following functions are nevertheless monotones:
\begin{equation}
    \label{monotone_ansatz1}    \sup_\rho d(\chan \otimes \iden (\rho), \chan_\mathrm{twirl} \otimes \iden (\rho)),
\end{equation}
and
\begin{equation}
    \label{monotone_ansatz2}    \sup_\rho d(\chan \otimes \iden (\rho), \chan_{\trev_A} \otimes \iden (\rho)).
\end{equation}
Indeed, both functions vanish exactly when $\chan = \chan_{\trev_A}$ and are non-increasing under arbitrary deterministic superchannels, not only under T-symmetric ones. Note, however, that the analogue of Theorem~\ref{thm:twirl} need not hold.

For example, choosing the diamond distance gives the monotone
\begin{equation}
    \label{diamond} || \chan - \chan_\mathrm{twirl} ||_\diamond = \frac{1}{2} || \chan - \chan_{\trev_A} ||_\diamond.
\end{equation}
This quantity can be computed by semidefinite programming in general, and an analytic expression is known for unitary channels \cite{aharonov_quantum_1998}. However, as shown below, comparison with the Choi-state monotone \eqref{monotone_trs} reveals that the diamond-distance monotone \eqref{diamond} does not resolve the convertibility structure sufficiently well.

To illustrate this point, consider the diamond-distance monotone in \eqref{diamond} for a qubit system whose motion-reversal operator is $\conj$. Using the formula for the diamond distance between unitary channels \cite{aharonov_quantum_1998}, if the eigenvalues of $\overline{U}U$ are $e^{\pm 2 i \phi}$, then
\begin{align}
    \frac{1}{2} || \adj_U - \adj_{U_\conj} ||_\diamond = |\sin 2 \phi|.
\end{align}
This quantity attains its maximal value, $1$, precisely at $\phi=\pi/4$.
For two copies of the unitary channel, defined on a two-qubit system with motion-reversal operator $\conj \otimes \conj$, one instead obtains
\begin{align}
    \frac{1}{2} || \adj_{U \otimes U} - \adj_{(U \otimes U)_{\conj \otimes \conj}} ||_\diamond = \left\{ \begin{array}{lc}
        1 & \frac{\pi}{8} \leq \phi \leq \frac{3}{8} \pi \\
        |\sin 4 \phi| & \text{otherwise.}
    \end{array} \right.
\end{align}
Thus, the replicated channel attains the maximal diamond distance throughout the interval $\phi \in [\pi/8,3\pi/8]$. By contrast, the fidelity of the replicated channel reaches its minimum value, $0$, only when the original single-copy channel itself maximally violates T-symmetry. Consequently, the channel $\adj_{U \otimes U}$ at $\phi=\pi/8$ cannot be converted into the corresponding channel at $\phi=\pi/4$, because the latter has strictly smaller fidelity. The diamond-distance monotone does not detect this obstruction, since it is maximal in both cases. Fidelity therefore provides a finer resolution of the convertibility structure than the diamond distance.

This limited resolution is not a consequence of replacing the minimization in \eqref{monotone_general} with twirling as in \eqref{twirling}; rather, it reflects a general feature of channel distinguishability. The distinguishability of channels behaves differently from that of states under replication. For instance, nonorthogonal unitary channels can become perfectly distinguishable when a sufficiently large but finite number of copies is available. Consequently, resource measures based on channel distinguishability can readily attain their maximal value upon replication, even when each individual channel contains only a small amount of resource. For kinematic T-violation, however, no finite number of weakly resourceful objects can generate a maximal resource. The limited resolution thus arises from a mismatch between the properties of channel-distinguishability measures and those of the channel resource theory.

\subsection{Convertibility and golden units}\label{subsec:golden_unit}
Whereas resource monotones readily provide impossibility criteria, establishing that a conversion is possible generally requires the explicit construction of a T-symmetric superchannel and is therefore substantially more difficult. We focus on specific conversion scenarios that highlight the properties of kinematic T-violation. A general analytical characterization of convertibility remains open, although individual conversion problems can, in principle, be formulated as semidefinite programs.

We write $\chan_1 \succ \chan_2$ when $\chan_1$ can be converted into $\chan_2$ by a T-symmetric superchannel, and $\chan_1 \nsucc \chan_2$ when no such conversion is possible.

We first introduce a canonical form for unitary channels under invertible T-symmetric superchannels. The key ingredient is unitary congruence (see Sec.~4.4 of \cite{horn_matrix_2012}, or \cite{wigner_normal_1960} for this version): for any $d$-dimensional unitary matrix $U$, there exists a unitary $E$ such that
\begin{align}
    \label{congruence2}  E^\top U E &= \id_{d-2q} \oplus W_{\phi_1} \oplus \cdots \oplus W_{\phi_q}, \\
    W_{\phi_j} &:= \left[ \begin{array}{cc}
        0 & e^{i \phi_j} \\
        e^{-i \phi_j} & 0
    \end{array}\right],
\end{align}
where $e^{\pm 2 i \phi_j}$, with $\phi_j \in (0,\pi/2]$, are the eigenvalues of $\overline{U}U$ that differ from $1$. The superchannel with encoder $\adj_E$ and decoder $\adj_{E^\top}$ is invertible and T-symmetric. Therefore, $\adj_U$ and $\adj_{E^\top U E}$ are interconvertible:
\begin{equation}
    \adj_U \sim \adj_{E^\top U E}.
\end{equation}
We therefore call the right-hand side of \eqref{congruence2} the canonical form of $U$. It also shows that convertibility between unitary channels $\adj_U$ and $\adj_V$ is determined solely by the eigenvalues of $\overline{U}U$ and $\overline{V}V$. For a symmetric unitary satisfying $U^\top=U$, no $W$ blocks appear: the identity block $\id_{d-2q}$ in \eqref{congruence2} does not contribute to T-violation.

The resource theories of $\mathbb{Z}_2$ asymmetry and imaginarity \cite{hickey_quantifying_2018} admit a universal golden unit: a state that can be converted into any other state, in any dimension, by free operations. Kinematic T-violation shares this property.
\begin{theorem}\label{thm:golden_unit}
    Let $G$ be defined by
    \begin{equation}
        G := W_{\pi/4} = \left[ \begin{array}{cc}
        0 & e^{\pi i/4} \\
        e^{- \pi i/4} & 0
        \end{array}\right].
    \end{equation}
    Then $\adj_G$ can be converted into any unital channel by a T-symmetric superchannel.
\end{theorem}
\noindent The proof is presented in Appendix~\ref{sec:golden_unit}.
The existence of a universal golden unit may seem natural because T-symmetry resembles a $\mathbb{Z}_2$ symmetry: it can be expressed as invariance of the Choi operator under $\swapio$. However, Choi operators satisfy additional dynamical constraints that break the isotropy of the ambient state space. Theorem~\ref{thm:golden_unit} therefore does not follow directly from the corresponding result for $\mathbb{Z}_2$ asymmetry.

The existence of a finite-dimensional universal golden unit constrains the quantitative behavior of kinematic T-violation. In the terminology of \cite{coecke_mathematical_nodate}, this resource is neither quantity-like nor quality-like. Specifically, a resource theory is called quantity-like when $a_1 \otimes a_2 \simeq b_1 \otimes b_2$ and $a_1 \preceq b_1$ together imply $a_2 \succeq b_2$. Kinematic T-violation does not satisfy this condition because the universal golden unit acts as a \textit{catalyst}:
\begin{equation}
    \adj_G \otimes \chan \simeq \adj_G \otimes \chan',
\end{equation}
for every pair of unital channels $\chan$ and $\chan'$.

A resource theory is called quality-like when $a \otimes a \simeq a$ holds for every resource $a$ \cite{coecke_mathematical_nodate}. For kinematic T-violation, $\chan \otimes \chan \succ \chan$, because discarding one copy is a T-symmetric superchannel. Conversely, multiplicativity of the fidelity, Eq.~\eqref{multiplicativity}, implies
\begin{equation}
    1 > \fidelity^\trev (\chan) > 0 \quad \Rightarrow \quad \chan \nsucc \chan \otimes \chan.
\end{equation}
Therefore, kinematic T-violation is not quality-like. The universal golden unit is an exception to this obstruction to replication:
\begin{equation}
    \adj_{G} \succ \adj_{G} \otimes \adj_{G},
\end{equation}
where $G$ is defined in \eqref{golden}.

There also exists a classical golden unit: a doubly stochastic map that can be converted into any unital channel by a T-symmetric superchannel. Define the doubly stochastic map $\circlearrowleft$ on $X_3 = \{ 0,1,2 \}$ by the matrix
\begin{equation}
    \label{circulator}  \left[\begin{array}{ccc}
        0 & 1 & 0 \\
        0 & 0 & 1 \\
        1 & 0 & 0
    \end{array}\right].
\end{equation}
This matrix represents a three-port circulator in scattering theory and is a prototypical nonreciprocal scattering matrix. With motion reversal defined as the identity on the index set $X_3$, $\circlearrowleft$ can be converted into the golden unit $\adj_G$ by a T-symmetric superchannel (see Appendix~\ref{sec:golden_unit}). Thus, $\circlearrowleft$ is also a golden unit.

A distinction between classical and quantum T-violation emerges when the dimension is fixed: no two-dimensional doubly stochastic map is a golden unit. Without a dimensional restriction, however, kinematic T-violation is not intrinsically quantum, in the sense that every unital T-violating channel can be generated from the doubly stochastic map $\circlearrowleft$.

\section{Conclusion and outlook}\label{sec:conclusion}
We have established a rigorous framework for quantifying T-violation in quantum channels using operational resource theory. Our definition of T-symmetric channels draws on both the conventional treatment of motion-reversal operators and the operational formulation of time reversal by Oreshkov and Cerf \cite{oreshkov_operational_2015}. In this framework, T-violation is characterized by the difference between a channel and its transpose, up to unitary transformations determined by the chosen motion-reversal operator.

Overall T-violation comprises two contributions: kinematic T-violation, which depends on the motion-reversal operator, and nonunitality, associated with channel non-invertibility. We distinguish these contributions structurally by identifying nested classes of resource-nongenerating operations, or free superchannels, for T-violation and nonunitality. Quantitatively, we prove that the relative-entropy measure of overall T-violation decomposes exactly into the corresponding measure of nonunitality and the kinematic T-violation of the closest unital channel. This identity provides a rigorous separation of the two contributions.

This quantitative decomposition relies on two techniques with implications beyond the present resource theory. The first is a nonstandard construction of dynamical resource monotones based on Choi states. These monotones can resolve distinctions in the convertibility of kinematic T-violation resources that channel-distinguishability-based monotones fail to detect. Our analysis thus illustrates a limitation of channel distinguishability as a basis for quantifying dynamical resources.

The second technique applies the generalized Pythagorean theorem for statistical manifolds \cite{matsuda_information_2022} to nested resource theories. In the resulting information-geometric picture, overall T-violation, nonunitality, and kinematic T-violation are represented by relative-entropy divergences associated with the three sides of a right triangle. Relations of this kind have been established, for example, between coherence and imaginarity, for which closed-form expressions for the resource monotones are available. Our analysis extends this relation to other nested state resource theories, without requiring closed-form expressions for the monotones.

We have also investigated resource monotones and convertibility for kinematic T-violation. The dynamical resource-theoretic framework permits comparisons between distinct physical systems, even when their motion-reversal operators differ. In particular, systems of the same dimension equipped with conjugation or skew-conjugation motion-reversal operators have isomorphic convertibility structures.

The fidelity-based monotone is particularly useful because it is multiplicative under parallel composition. This property implies that no finite number of resources with nonzero fidelity can be combined and converted into a resource with zero fidelity, i.e., a maximally T-violating channel. It also rules out exact replication of T-violating channels with intermediate fidelity by T-symmetric superchannels, while leaving replication of the golden unit possible.

The golden unit reveals further structural properties of kinematic T-violation. As a maximal resource, it can be converted into any unital channel by a T-symmetric superchannel. Its catalytic role shows that kinematic T-violation is not a quantity-like resource in the sense of \cite{coecke_mathematical_nodate}; nevertheless, the resource is not quality-like either. We also identify a classical golden unit: a doubly stochastic matrix that can be converted into any unital channel, including quantum channels, by a T-symmetric superchannel. Thus, kinematic T-violation is not an exclusively quantum resource.

Beyond the open questions discussed in the main text, several directions merit further investigation. One is to extend the framework to probabilistic operations; another is to characterize T-violation arising from the sequential composition of noncommuting T-symmetric channels. In connection with the latter, formulating T-violation for $S$-matrices---composed via the Redheffer star product \cite{Redheffer,catalano2026quantum}---holds potential applications in scattering theory and wave reciprocity. The dynamical resource-theoretic principles developed here provide a starting point for these extensions.

\begin{acknowledgments}
    We thank many people for discussions, and especially Jacopo Surace, Nelly Ng, Satoshi Yoshida, Philip Taranto, and Kim Worrall for their valuable feedback. We acknowledge the use of Google Gemini (3.6 Thinking) and OpenAI Prism for proofreading and language polishing. The authors also utilized Google Gemini (3.6 Thinking) to assist in conceptual discussions and in drafting proof ideas for Appendix~\ref{sec:twirl}. All AI-generated outputs, including mathematical arguments and manuscript text, were independently verified, revised, and approved by the authors, who take full responsibility for the final proofs, scientific claims, and overall presentation. This work was supported by MEXT Quantum Leap Flagship Program (MEXT QLEAP) JPMXS0118069605 and JPMXS0120351339; Japan Science and Technology Agency (JST) as part of Adopting Sustainable Partnerships for Innovative Research Ecosystem (ASPIRE), Grant Number JPMJAP25A3; JST CREST, Grant Number JPMJCR25I5; JST NEXUS, Grant Number JPMJNX26C9; JSPS KAKENHI Grant No. 21H03394 and No. 23K21643; and IBM Quantum.
    K.K.\ was supported by a Mike and Ophelia Lazaridis Fellowship, a Funai Overseas Scholarship, and a Perimeter Residency Doctoral Award. 
\end{acknowledgments}

\appendix


\section{Partitioned process theory}\label{sec:partition}
We introduce the categories of deterministic, unital, and T-symmetric superchannels, denoted by $\cat_\mathrm{det}$, $\cat_\unital$, and $\cat_\trs$, respectively, and establish the subcategory relations
\begin{equation}
    \label{subcategory} \cat_\trs \hookrightarrow \cat_\unital \hookrightarrow \cat_\mathrm{det}.
\end{equation}
We can then define three partitioned process theories \cite{coecke_mathematical_nodate}, each consisting of a category and a subcategory, corresponding to the following resource theories:
\begin{align}
    (\cat_\mathrm{det}, \cat_\unital) \qquad &\text{nonunitality}\\
    (\cat_\unital, \cat_\trs) \qquad &\text{kinematic T-violation}\\
    (\cat_\mathrm{det}, \cat_\trs) \qquad &\text{overall T-violation}.
\end{align}

\begin{definition}\label{def:category_deterministic}
    The category $\cat_\mathrm{det}$ of deterministic superchannels is defined by the following ingredients:
    \begin{itemize}
        \item An object $A := (\hil_A, \trev_A)$ is a pair consisting of a Hilbert space $\hil_A$ and an antiunitary operator $\trev_A$.
        \item An arrow $f: A \multimap B$ is a deterministic quantum superchannel from $\B(\B(\hil_A),\B(\hil_A))$ to $\B(\B(\hil_B),\B(\hil_B))$.
        \item The parallel composition of objects $A = (\hil_A, \trev_A)$ and $B = (\hil_B, \trev_B)$ is given by $A \otimes B := (\hil_A \otimes \hil_B, \trev_A \otimes \trev_B)$. The parallel composition of arrows $f:A \multimap C$ and $g:B \multimap D$ is also given by the tensor-product superchannel $f \otimes g$ from $\B(\B(\hil_A \otimes \hil_B),\B(\hil_A \otimes \hil_B))$ to $\B(\B(\hil_C \otimes \hil_D),\B(\hil_C \otimes \hil_D))$.
        \item The void object $I$ is defined as $(\mathbb{C},\ast)$, where $\ast$ denotes complex conjugation on scalars.
    \end{itemize} 
\end{definition}
The category axioms for $\cat_\mathrm{det}$ follow from the closure of deterministic superchannels under associative sequential composition. The identity arrow is the trivial superchannel implemented by the identity encoder--decoder pair. Moreover, $(\cat_\mathrm{det},\bullet,\otimes,I)$ is a symmetric monoidal category, with the monoidal structure inherited from the usual tensor product. Although it is not used in the subsequent analysis, Fig.~\ref{fig:ppt} summarizes the diagrammatic elements of $\cat_\mathrm{det}$ and their corresponding circuit representations.
\begin{figure*}[tbp]
    \includegraphics[width=.8\textwidth]{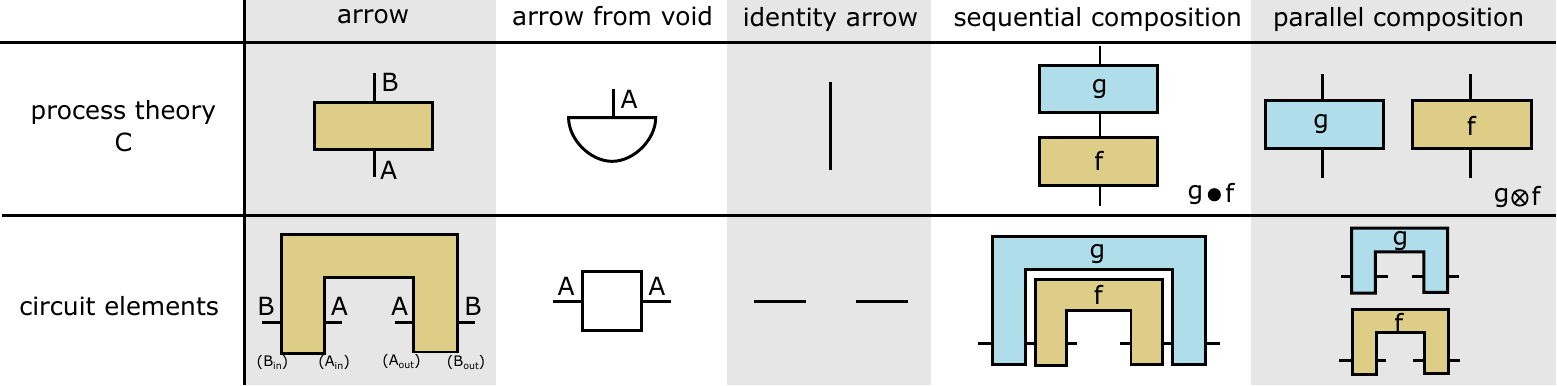}
    \caption{\label{fig:ppt} Correspondence between elements of the process theories ($\cat_\mathrm{det}$, $\cat_\unital$, and $\cat_\trs$) and elements of quantum circuits. The circuit element shaped like ``$\sqcap$'' represents a one-slot quantum comb.}
\end{figure*}

Ordinary quantum channels appear in $\cat_\mathrm{det}$ as arrows from the void object. Indeed, when a superchannel has no input channel, its encoder and decoder simply compose to form a single quantum channel. This reduction from superchannels to channels is used in the ``top-down'' approach to dynamical resource theories in \cite{ji_entropic_2023}.

The antiunitary operators assigned to objects are intended to represent motion reversal, although they have not yet entered the preceding definitions. The categories of deterministic and unital superchannels could be defined without these operators; nevertheless, we include them so that the category of T-symmetric superchannels can be realized as a subcategory of both. Definition~\ref{def:category_deterministic} also fixes a nontrivial convention: under parallel composition, motion-reversal operators compose by the ordinary tensor product $\trev_A \otimes \trev_B$. This convention is inherited by the subcategories introduced below.

As shown in the main text, the classes of unital and T-symmetric superchannels are each closed under sequential composition ($\bullet$) and parallel composition ($\otimes$). Restricting the arrows of $\cat_\mathrm{det}$ to these respective classes therefore yields symmetric monoidal subcategories $\cat_\unital$ and $\cat_\trs$.
\begin{definition}\label{def:category_unital}
    The category of unital superchannels $\cat_\unital$ has the same objects as $\cat_\mathrm{det}$ and unital superchannels as arrows. The category of T-symmetric superchannels $\cat_\trs$ has the same objects as $\cat_\mathrm{det}$ and T-symmetric superchannels as arrows.
\end{definition}
The monoidal tensor $\otimes$ and the void object $I$ are inherited from $\cat_\mathrm{det}$. 

By definition, $(\cat_\mathrm{det},\cat_\unital)$ and $(\cat_\mathrm{det},\cat_\trs)$ both form partitioned process theories, respectively defining nonunitality and T-violation.
Additionally, \eqref{superchannel_subclass} implies
\begin{equation}
   \cat_\trs \hookrightarrow \cat_\unital,
\end{equation}
that is, $\cat_\trs$ is a subcategory of $\cat_\unital$. Therefore, the pair $(\cat_\unital,\cat_\trs)$ also forms a partitioned process theory, defining kinematic T-violation.


\section{Unital superchannels are CPTP maps on Choi states}\label{sec:proof_choi_CPTP}
This appendix proves Lemma \ref{lem:choistate}.
In terms of the Choi operator $\choi[f]$ of $f$, this map is represented as
\begin{equation}
    \widetilde{f}(\cstate[\chan]) = \frac{1}{d_B} \choi[f] \star d_A \cstate[\chan] = \trace_A \left[ \frac{d_A}{d_B} \choi[f] \cstate[\chan]^\top \right].
\end{equation}
It therefore suffices to show that $d_A \choi[f]/d_B$ is the Choi operator of a CPTP map from $\B(\hil_A \otimes \hil_A)$ to $\B(\hil_B \otimes \hil_B)$. Complete positivity follows from positive-semidefiniteness of $\choi[f]$; trace preservation is equivalent to
\begin{equation}
    \trace_B \left[ \frac{d_A}{d_B} \choi[f] \right] = \id_{\hil_A} \otimes \id_{\hil_A}.
\end{equation}
To verify this condition, label the input and output spaces by $\hil_A^\mathrm{in}$, $\hil_A^\mathrm{out}$, $\hil_B^\mathrm{in}$, and $\hil_B^\mathrm{out}$. Since $\choi[f]$ is the Choi operator of a deterministic superchannel, we have \cite{chiribella_theoretical_2009}
\begin{equation}
    \label{trace_rightmost}    \trace_{\hil_B^\mathrm{out}} [ \choi[f] ] = \frac{1}{d_A} \trace_{\hil_B^\mathrm{out}, \hil_A^\mathrm{out}} [ \choi[f] ] \otimes \id_{\hil_A^\mathrm{out}}.
\end{equation}
Since $f$ is unital, its adjoint $f^\dagger$ is also a unital superchannel (Lemma~\ref{lem:unital}). Its Choi operator is given by
\begin{equation}
    \choi[f^\dagger] = \swapio(\overline{\choi[f]}),
\end{equation}
where the overline denotes complex conjugation in the computational basis. Since $f^\dagger$ is a deterministic superchannel, we have
\begin{align}
    \trace_{\hil_B^\mathrm{in}} [ \choi[f] ] &= \overline{\trace_{\hil_B^\mathrm{in}} [ \swapio(\choi[f^\dagger]) ]} \\
    &= \swapio \left( \overline{\trace_{\hil_B^\mathrm{out}} [ \choi[f^\dagger] ]} \right) \\
    &= \swapio \left( \frac{1}{d_A} \overline{ \trace_{\hil_B^\mathrm{out}, \hil_A^\mathrm{out}} [ \choi[f^\dagger] ] \otimes \id_{\hil_A^\mathrm{out}}} \right) \\
    \label{trace_leftmost}    &= \frac{1}{d_A}  \swapio \left( \overline{\trace_{\hil_B^\mathrm{out}, \hil_A^\mathrm{out}} [ \choi[f^\dagger] ]} \right) \otimes \id_{\hil_A^\mathrm{in}},
\end{align}
where the action of $\swapio$ on partially traced systems is understood as relabeling the remaining system.
Combining \eqref{trace_rightmost} and \eqref{trace_leftmost}, we obtain
\begin{equation}
    \trace_{\hil_B^\mathrm{in}, \hil_B^\mathrm{out}} \left[ \frac{d_A}{d_B} \choi[f] \right] = r \id_{\hil_A} \otimes \id_{\hil_A},
\end{equation}
for some positive number $r$ that determines the uniform trace scaling. Since $\widetilde{f}$ maps Choi states to Choi states, one must have $r=1$. This completes the proof of Lemma~\ref{lem:choistate}.

\section{2-dimensional channel spaces are \textit{e}-TGC}\label{sec:flat}
In this section, we show that the spaces $S_\unital$ and $S_\trs$ defined by
\begin{align}
    &S_\unital := \left\{ \cstate[\chan_\unital] \middle| \chan_\unital \in \unital(A) \right\} \\
    &= \left\{ \rho \in \density \middle| \trace_{A_\mathrm{in}} [\rho] = \frac{\id_{A_\mathrm{out}}}{d_A}, ~ \trace_{A_\mathrm{out}} [\rho] = \frac{\id_{A_\mathrm{in}}}{d_A} \right\}, \\
    &S_\trs := \left\{ \cstate[\chan_\trs] \middle| \chan_\trs \in \trs(A) \right\} \\
    &= \left\{ \rho \in S_\unital \middle| \adj_{X^\dagger_\mathrm{out} \otimes X^\top_\mathrm{in}} \circ \swapio (\rho) = \rho  \right\},
\end{align}
where $\density$ denotes the space of density operators on $\hil_A^\mathrm{in} \otimes \hil_A^\mathrm{out}$, are both \textit{e}-TGC when $d_A = \dim \hil_A = 2$.

Exceptional features of qubit systems used here are that any qubit unital channel is a mixture of unitary channels \cite{LANDAU1993107}, and that any such channel can be written as a mixture of four mutually orthogonal unitary channels \cite{BETHRUSKAI2002159}. The latter implies that any positive affine combination $\sum_i p_i \ketbra[\psi_i]$ ($p_i > 0$) of (possibly non-orthogonal) maximally entangled states $\ket[\psi_i]$ has a decomposition
\begin{equation}
    \sum_i p_i \ketbra[\psi_i] = \sum_{j=0}^3 q_j \ketbra[\phi_j],
\end{equation}
where $\{ \ket[\phi_j] \}_{j=0,1,2,3}$ is a set of orthonormal maximally entangled vectors.

Let us start with $S_\unital$. Let $\rho = \sum_{j=0}^3 q_j \ketbra[\phi_j]$ and $\rho' = \sum_{j=0}^3 q'_j \ketbra[\phi'_j]$ be such orthogonal decompositions of $\rho, \rho' \in S_\unital$. We then have
\begin{align}
    &t \log \rho + (1-t) \log \rho' \\
    &= \sum_{j=0}^3 t \log q_j \ketbra[\phi_j] +  \sum_{j=0}^3 (1-t) \log q'_j \ketbra[\phi'_j].
\end{align}
Since the right-hand side is a linear combination of maximally entangled states with non-positive coefficients, it has a decomposition $\sum_j - q''_j \ketbra[\phi''_j]$ into another set of orthonormal maximally entangled vectors $\{ \ket[\phi''_j] \}$ with $q''_j \geq 0$. This leads to
\begin{equation}
    \exp[t \log \rho + (1-t) \log \rho'] = \sum_j e^{- q''_j} \ketbra[\phi''_j],
\end{equation}
which, after normalization, is contained in $S_\unital$. This proves that any \textit{e}-geodesic in $S_\unital$ is contained in $S_\unital$, and hence that $S_\unital$ is \textit{e}-TGC.

To see the \textit{e}-TGC of $S_\trs$, note that the constraint $\adj_{X^\dagger_\mathrm{out} \otimes X^\top_\in} \circ \swapio (\rho) = \rho$, imposed on $S_\unital$ to define $S_\trs$, is affine in \textit{e}-coordinates.
More directly, if $\rho, \rho' \in S_\unital$ satisfy a symmetry $U \rho U^\dagger = \rho$ and $U \rho' U^\dagger = \rho'$, so do $\log \rho$, $\log \rho'$, $t \log \rho + (1-t) \log \rho'$, and $\exp[t \log \rho + (1-t) \log \rho']$. The \textit{e}-TGC of $S_\trs$ follows by taking $\adj_U = \adj_{X^\dagger_\mathrm{out} \otimes X^\top_\in} \circ \swapio$.

\section{Optimality of the twirled channel}\label{sec:twirl}
This appendix proves Theorem~\ref{thm:twirl}. The relative-entropy case follows directly from Proposition~2 of \cite{gour_measuring_2009}.

In general, RPS of a channel is represented by the unitary transformation in \eqref{choi_time_reversal} acting on the corresponding Choi state. Hence, $d(a,b)$ is invariant under simultaneously applying RPS to $a$ and $b$. In particular, for any unital channel $\chan$ and any T-symmetric channel $\chan_\mathrm{free}$, we have
\begin{equation}
    \label{distance_invariance} d(\cstate[\chan],\cstate[\chan_\mathrm{free}]) = d(\cstate[\chan_\trev],\cstate[\chan_\mathrm{free}]).
\end{equation}

First, consider the trace distance $d(a,b) = | a-b |_1$. By \eqref{distance_invariance},
\begin{align}
    |\cstate[\chan] - \cstate[\chan_\mathrm{free}]|_1
    &= \frac{|\cstate[\chan] - \cstate[\chan_\mathrm{free}]|_1}{2} + \frac{|\cstate[\chan_\trev] - \cstate[\chan_\mathrm{free}]|_1}{2}\\
    & \geq \frac{|\cstate[\chan] - \cstate[\chan_\trev]|_1}{2} \\
    & = | \cstate[\chan] - \cstate[\chan_\mathrm{twirl}] |_1 ,
\end{align}
for every T-symmetric channel $\chan_{\mathrm{free}}$, where the second line follows from the triangle inequality. Thus, the trace distance is minimized by $\chan_\mathrm{twirl}$.

For the squared Frobenius norm $d(a,b) = || a-b ||^2$, define $L := (\cstate[\chan] - \cstate[\chan_\trev])/2$ and expand $|| \cstate[\chan] - \cstate[\chan_\mathrm{free}] ||^2$ as
\begin{align}
    & \trace[ (\cstate[\chan_\mathrm{twirl}] + L - \cstate[\chan_\mathrm{free}])^2 ] \\
    & = \trace[ (\cstate[\chan_\mathrm{twirl}] - \cstate[\chan_\mathrm{free}])^2 ] + \trace[L^2] \\
    &\quad + 2\trace[ (\cstate[\chan_\mathrm{twirl}] - \cstate[\chan_\mathrm{free}]) L].
\end{align}
The last term vanishes because, under the unitary transformation representing RPS on Choi-operators, $\cstate[\chan_\mathrm{twirl}] - \cstate[\chan_\mathrm{free}]$ is invariant whereas $L$ changes sign. Therefore,
\begin{equation}
    || \cstate[\chan] - \cstate[\chan_\mathrm{free}] ||^2 = || \cstate[\chan_\mathrm{twirl}] - \cstate[\chan_\mathrm{free}] ||^2 + \trace[L^2],
\end{equation}
which is minimized if and only if $\cstate[\chan_\mathrm{twirl}] = \cstate[\chan_\mathrm{free}]$.

\section{Convexity and the range of relative entropy monotones}\label{sec:convexity}
We derive the convexity in \eqref{convexity} and the range in \eqref{range} of the relative-entropy monotones. These properties are not specific to T-violation, but hold for a broad class of relative-entropy monotones.

We first establish convexity. The twirling operation is affine in the sense that
\begin{equation}
    (p \chan + q \chan')_\mathrm{twirl} = p \chan_\mathrm{twirl} + q \chan'_\mathrm{twirl}, 
\end{equation}
and this affine structure is inherited directly by the Choi-state representation. In general, if the state transformation $\mathcal{G} : \rho \mapsto \mathcal{G} (\rho)$ (which need not be CPTP) is affine, then
\begin{align}
    & S(p \rho + (1-p) \rho' || \mathcal{G}(p \rho + (1-p) \rho') ) \\
    &= S(p \rho + (1-p) \rho' || p \mathcal{G}(\rho) + (1-p)\mathcal{G} (\rho') ) \\
    & \leq p S(\rho || \mathcal{G}(\rho)) + (1-p) S(\rho' || \mathcal{G}(\rho')),
\end{align}
by the joint convexity of relative entropy. Thus, the relative-entropy monotone is convex whenever the closest free state depends affinely on the input state.

Having established convexity, it suffices to consider pure input states to prove the range in \eqref{range}. For pure inputs (i.e., unitary channels), Eq.~\eqref{entropy_difference} reduces to
\begin{equation}
    \label{twirled_pure}    S(\cstate[\chan_\mathrm{twirl}]) = S \left( \frac{\cstate[\chan] + \cstate[\chan_{\trev_A}]}{2} \right).
\end{equation}
For pure states $\ket[\psi]$ and $\ket[\phi]$, we have
\begin{equation}
    \label{binary}  S \left( \frac{\ketbra[\psi] + \ketbra[\phi]}{2} \right) = h_2 \left( \frac{1+ | \braket{\psi}{\phi} |}{2} \right),
\end{equation}
where $h_2$ is the binary entropy. Since this value lies in $[0,1]$, so does \eqref{twirled_pure}; the range in \eqref{range} then follows from convexity.

Equation~\eqref{binary} also yields the relation between relative entropy and fidelity in \eqref{relative_entropy_fidelity}.

\section{Golden unit}\label{sec:golden_unit}
In this appendix, we prove Theorem~\ref{thm:golden_unit} by explicitly constructing the Choi operators of the required T-symmetric superchannels. We also construct a T-symmetric superchannel that converts the three-port circulator $\circlearrowleft$ into $\adj_G$.

Consider the conversion from $\adj_G$ to an arbitrary unital channel $\chan \in \unital(B)$, where $\dim \hil_B = d$.
Define an operator $J_\chan$ on $\hil_B^\mathrm{in} \otimes \hil_B^\mathrm{out} \otimes \hil_A^\mathrm{in} \otimes \hil_A^\mathrm{out}$, with $\dim \hil_A = 2$, by
\begin{equation}
    J_\chan := D + J_+ + J_-,
\end{equation}
where
\begin{align}
    D &:= \frac{1}{d} \id_{B_\mathrm{in}} \otimes \id_{B_\mathrm{out}} \otimes \frac{\ketbra[00] + \ketbra[11]}{2},\\
    J_+ &:= \frac{\choi[\chan] + \choi[\chan^\top]}{2} \otimes \frac{\ketbra[01] + \ketbra[10]}{2}, \\
    J_- &:= \frac{\choi[\chan] - \choi[\chan^\top]}{2i} \otimes \frac{\ket[10] \bra[01] - \ket[01] \bra[10]}{2}.
\end{align}
We first show that $J_\chan$ is the Choi operator of a deterministic superchannel for every unital channel $\chan$.
To establish positive semidefiniteness, we use the expression
\begin{equation}
    J_+ + J_- = \choi[\chan] \otimes \ketbra[\lambda_1] + \choi[\chan^\top] \otimes \ketbra[\lambda_2] \geq 0,
\end{equation}
where
\begin{align}
    \ket[\lambda_1] = \frac{\ket[01] - i \ket[10]}{\sqrt{2}}, \quad \ket[\lambda_2] = \frac{i \ket[01] + \ket[10]}{\sqrt{2}}.
\end{align}
Since $D$ and $J_+ + J_-$ are both positive semidefinite, so is $J_\chan$.
To verify the no-signaling condition for the comb, note that
\begin{equation}
    \trace_{B_\mathrm{out}}[ \choi[\chan^\top] ] = \id_{B_\mathrm{in}},
\end{equation}
because $\chan$ is unital. It follows that
\begin{align}
    \trace_{B_\mathrm{out}}[J_\chan] &= \frac{1}{2} \id_{B_\mathrm{in}} \otimes \id_{A_\mathrm{in}} \otimes \id_{A_\mathrm{out}}, \\
    \trace_{B_\mathrm{out}, A_\mathrm{in}, A_\mathrm{out}}[J_\chan] &= 2 \id_{B_\mathrm{in}}.
\end{align}
Therefore, $J_\chan$ is the Choi operator of a deterministic superchannel.

Moreover, $J_\chan$ defines a T-symmetric superchannel. Since $\swapio (\choi[\chan]) = \choi[\chan^\top]$, each of the three components $D$, $J_+$, and $J_-$ is invariant under the input-output swap $\swapio$.

Finally, a direct calculation gives
\begin{align}
    \trace_A [J_\chan (\id_B \otimes \choi[\adj_G]^\top)] = \choi[\chan].
\end{align}
Thus, the channel $\adj_G$ can be converted into the unital channel $\chan$ by the T-symmetric superchannel defined by $J_\chan$.

Now let $A$ be the classical system with index set $X_3 = \{ x_0,x_1,x_2 \}$ and motion reversal given by the identity, and let $B$ be a two-dimensional quantum system. Consider the following Choi operator on $\hil_B^\mathrm{in} \otimes \hil_B^\mathrm{out} \otimes \hil_A^\mathrm{in} \otimes \hil_A^\mathrm{out}$:
\begin{align}
    & \frac{1}{6} \id_{B_\mathrm{in} \otimes B_\mathrm{out}} \otimes \sum_{k=0}^2 x_k \otimes x_k + \frac{2}{3} \sum_{l \neq m} \ketbra[\psi_{l,m}] \otimes x_l \otimes x_m, \\
    & \left( \ket[\psi_{l,m}] := \frac{\ket[0]_{B_\mathrm{in}} \ket[1]_\mathrm{out} + i \epsilon_{ml} \ket[1]_{B_\mathrm{in}} \ket[0]_\mathrm{out} }{\sqrt{2}} \right)
\end{align}
where $\epsilon_{ml}$ is a Levi-Civita-like symbol. A direct calculation verifies that this operator defines a T-symmetric superchannel from $A$ to $B$ that converts $\circlearrowleft$ into $\adj_G$.

\bibliography{ref_arxiv.bib}
\end{document}